\documentclass[sigconf, nonacm=true]{acmart}

\newcommand{\AI}{$\mathsf{SynchPack}\text{-}3$}
\newcommand{\AII}{$\mathsf{SynchPack}\text{-}1$}
\newcommand{\AIII}{$\mathsf{SynchPack}\text{-}2$}
\newcommand{\Tetp}{$\mathsf{Tetris}\text{-}\mathsf{p}$}
\newcommand{\Tetnp}{$\mathsf{Tetris}\text{-}\mathsf{np}$}
\newcommand{\PSRS}{$\mathsf{PSRS}$}
\newcommand{\JSQMW}{$\mathsf{JSQ}\text{-}\mathsf{MW}$}

\usepackage{booktabs}
\usepackage{float}
\usepackage{amssymb}
\usepackage{hyperref}
\usepackage{appendix}
\usepackage{enumitem}
\usepackage{amsmath}
\usepackage{amsthm}
\usepackage{amssymb}
\usepackage{graphicx}
\usepackage{epstopdf}
\usepackage{latexsym}
\usepackage[dvips]{epsfig}
\usepackage{mathrsfs}
\usepackage{dsfont}
\usepackage{fix2col}
\usepackage{algorithm}
\usepackage{algpseudocode}
\usepackage{color}
\usepackage{subcaption}
\usepackage{caption}
\usepackage{appendix}

\DeclareMathOperator*{\argmin}{arg\,min}

\newcommand{\be}{\begin{eqnarray}}
	
	\newcommand{\ee}{\end{eqnarray}}
\newcommand{\ben}{\begin{eqnarray*}}
	\newcommand{\een}{\end{eqnarray*}}
\newcommand{\bfl}{\begin{flalign*}}
	\newcommand{\efl}{\end{flalign*}}
\newcommand{\dref}[1]{(\ref{#1})}
\newcommand{\expect}[1]{{\mathbb E} \Bigl[ #1\Bigr]}

\newcommand{\calG}{{\mathcal G}}

\newcommand{\calJ}{{\mathcal J}}

\newcommand{\calK}{{\mathcal K}}
\newcommand{\calS}{{\mathcal S}}

\newcommand{\calM}{{\mathcal M}}

\newcommand{\calT}{{\mathcal T}}
\newcommand{\calV}{{\mathcal V}}

\theoremstyle{remark}
\newtheorem{remark}{Remark}

\DeclareCaptionFormat{algor}{%
	\textbf{#1#2}#3\offinterlineskip} 
\DeclareCaptionStyle{algori}{singlelinecheck=off,format=algor,labelsep=space}
\newcounter{parentalgorithm}


\graphicspath{{Fig/}}

\acmDOI{}
\begin{document}
\title[Scheduling Parallel-Task Jobs Subject to Packing and Placement Constraints]{Scheduling Parallel-Task Jobs Subject to\\ Packing and Placement Constraints}           
                 
\author{Mehrnoosh~Shafiee and~Javad~Ghaderi}
\affiliation{%
  \institution{Department of Electrical Engineering\\
  	Columbia University}
}
\begin{abstract}
\label{abstract} 
Motivated by modern parallel computing applications, we consider the problem of scheduling parallel-task jobs with heterogeneous resource requirements in a cluster of machines. 
Each job consists of a set of tasks that can be processed in parallel, however, the job is considered completed only when all its tasks finish their processing, which we refer to as ``\textit{synchronization}'' constraint. Further, assignment of tasks to machines is subject to ``\textit{placement}'' constraints, i.e., each task can be processed only on a subset of machines, and processing times can also be machine dependent.
Once a task is scheduled on a machine, it  requires a certain amount of resource from that machine for the duration of its processing. 
A machine can process (``\textit{pack}'') multiple tasks at the same time, however the cumulative resource requirement of the tasks should not exceed the machine's capacity.

Our objective is to minimize the weighted average of the jobs' completion times. The problem, subject to synchronization, packing and placement constraints,  is NP-hard, and prior theoretical results only concern much simpler models. For the case that migration of tasks among the placement-feasible machines is allowed, we propose a preemptive algorithm with an approximation ratio of $(6+\epsilon)$. In the special case that only one machine can process each task, we design an algorithm with improved approximation ratio of $4$. Finally, in the case that migrations (and preemptions) are not allowed, we design an algorithm with an approximation ratio of $24$. 
Our algorithms use a combination of linear program relaxation and greedy packing techniques. We present extensive simulation results, using a real traffic trace, that demonstrate that our algorithms yield significant gains over the prior approaches.
\end{abstract}	

\maketitle

\section{Introduction}
\label{introduction}
Modern parallel computing frameworks (e.g. Hadoop~\cite{Hadoop}, Spark~\cite{ApacheSpark}) have enabled large-scale data processing in computing clusters. In such frameworks, the data is typically distributed across a cluster of machines and is processed in multiple stages. In each stage, a set of tasks are executed on the machines, and once all the tasks in the stage finish their processing, the job is finished or moved to the next stage.  
For example, in MapReduce~\cite{dean2008mapreduce}, in the map stage, each map task performs local computation on a data block in a machine and writes the intermediate data to the
disk. In the reduce stage,  each reduce task pulls intermediate data from different maps, merges them, and computes
its output. While the reduce tasks can start pulling data as map tasks finish, the actual computation by the reduce tasks can only start once all the map tasks are done and their data pieces are received. Further, the job is not completed unless all the reduce tasks finish. Similarly, in Spark~\cite{zaharia2016apache,ApacheSpark}, the computation is done in multiple stages. The tasks in a stage can run in parallel, however, the next stage cannot start unless the tasks in the preceding stage(s) are all completed. 

We refer to such constraints as \textit{synchronization} constraints, i.e., a stage is considered completed only when all its tasks finish their processing. 
Such synchronizations could have a significant impact on the jobs' latency in parallel computing clusters~\cite{ananthanarayanan2010reining,zaharia2008improving,kambatla2010asynchronous,zaharia2016apache,cheatham1996bulk}. Intuitively, an efficient scheduler should complete all the (inhomogeneous) tasks of a stage more or less around the same time, while prioritizing the stages of different jobs in an order that minimizes the overall latency in the system.

Another main feature of parallel computing clusters is that jobs can have diverse tasks and processing requirements. This has been further amplified by the increasing
complexity of workloads, i.e., from traditional batch jobs, to
queries, graph processing, streaming, and machine learning jobs, that all need to share the same cluster. The cluster manager (\textit{scheduler}) serves  the tasks of various jobs by reserving their requested resources (e.g. CPU,
memory, etc.). For example, in
Hadoop~\cite{HadoopYarn}, the resource manager reserves the tasks' resource requirements by launching ``\textit{containers}'' in machines. Each container reserves required resources for processing of a task. To improve the overall latency, we therefore need a scheduler that \textit{packs} as many tasks as possible in the machines, while retaining their resource requirements. 

In practice, there are further placement constraints for processing tasks on machines. For example, each task is preferred to be scheduled on one of the machines that has its required data block~\cite{dean2008mapreduce,ananthanarayanan2011disk} (a.k.a. data locality), otherwise processing can slow down due to data transfer. The data block might be stored in multiple machines for robustness and failure considerations. However, if all these machines are highly loaded, the scheduler might actually need to schedule the task in a less loaded machine that does not contain the data. 

Despite the vast literature on scheduling algorithms, the theoretical results are mainly based on simple models where each machine processes one task at a time, each job is a single task, or tasks can be processed on any machine arbitrarily (see Related Work in Section~\ref{relatedwork}). The objective of this paper is to design scheduling algorithms, with \textit{theoretical guarantees}, under the modern features of data-parallel computing clusters, namely,
\begin{itemize}[leftmargin=*]
	\item \textit{packing constraint}: each machine is capable of processing multiple tasks at a time subject to its capacity.
	\item \textit{synchronization constraint}: tasks that belong to the same stage (job) have a collective completion time which is determined by the slowest task in the collection.   
	\item \textit{placement constraint}: task's processing time is machine-dependent and task is typically preferred to be processed on a subset of machines (e.g. where its input data block is located). 
\end{itemize}

For simplicity, we consider one dimension for task resource requirement (e.g. memory). While task resource requirements are in general multi-dimensional (CPU, memory, etc.), it has been
observed that memory is typically the bottleneck resource~\cite{HadoopYarn,nitu2018working}.

Note that the scheduler can only make scheduling decisions for the stages that have been released from various jobs up to that point (i.e., those that their preceding stages have been completed). In our model, we use the terms stage and job interchangeability. Our objective is to minimize the weighted sum of completion times of existing jobs (stages) in the system, where weights can encode different priorities for the jobs (stages). Clearly minimization of the average completion time is a special case of this problem with equal weights. We consider both preemptive and non-preemptive scheduling. In a \textit{non-preemptive} schedule, a task cannot be preempted (and hence cannot be migrated among machines) once it starts processing on a machine until it is completed. In a \textit{preemptive} schedule, a task may be preempted and resumed later in the schedule, and we further consider two cases depending on whether migration of a task among machines is allowed or not. 
\subsection{Related Work}
\label{relatedwork}


Default cluster schedulers in Hadoop~\cite{HadoopCapacity, HadoopFair,zaharia2010delay} focus primarily
on fairness and data locality. 
Such schedulers can make poor scheduling decisions by not packing tasks well together, or having a task running long without enough parallelism with other tasks in the same job.
Several cluster schedulers have been proposed to improve job completion times, e.g.~\cite{grandl2016graphene,rasley2016efficient,schwarzkopf2013omega,jin2011bar,grandl2015multi,verma2015large,liu2016dependency,wang2016maptask,yekkehkhany2018gb}. However, they either do not consider all aspects of packing, synchronization, and data locality, or use heuristics which are not necessarily efficient. 

We highlight four relevant papers~\cite{grandl2015multi,verma2015large,wang2016maptask,yekkehkhany2018gb} here.
Tetris~\cite{grandl2015multi} is a scheduler that assigns scores to tasks based on Best-Fit bin packing and Shortest-Remaining-Time-First (SRPF) heuristic, and gives priority to tasks with higher scores. The data locality is encoded in scores by imposing a remote penalty to penalize use of remote resources. 
Borg~\cite{verma2015large} packs multiple tasks of jobs in machines from high to low priority, modulated by a round-robin scheme within a priority to ensure fairness across jobs. The scheduler considers data locality by assigning tasks to machines that already have the necessary data stored.
The papers~\cite{wang2016maptask} and~\cite{yekkehkhany2018gb} focus on single-task jobs and study the mean delay of tasks under a stochastic model where if a task is scheduled on one of the remote servers that do not have the input data, its average processing time will be larger, by a multiplicative factor, compared to the case that it is processed on a local server that contains the data. They propose algorithms based on Join-the-Shortest-Queue and Max-Weight (JSQ-MW) to incorporate data locality in load balancing. 
This model is generalized in~\cite{yekkehkhany2018gb} to more levels of data locality.
	  However, these models do not consider any task packing in servers or synchronization issue among multiple tasks of the same job.

From a theoretical perspective, our problem of scheduling parallel-task jobs with synchronization, packing, and placement constraints, can be seen as a generalization of the concurrent open shop (COS) problem~\cite{ahmadi2005coordinated}. 
Unlike COS, where each machine processes one task at a time and each task can be processed on a specific machine, in our model a machine can process (pack) multiple tasks simultaneously subject to its capacity, and there are further task placement constraints for assigning tasks to machines. 
Minimizing the weighted sum of completion times in COS, is known to be APX-hard~\cite{garg2007order}, with several $2$-approximation algorithms in~\cite{chen2007supply, garg2007order, leung2007scheduling,mastrolilli2010minimizing, sachdeva2013optimal,bansal2010inapproximability}.
There is also a line of research on the parallel tasks scheduling (PTS) problem~\cite{garey1975bounds}. In PTS, each job is only a single task that requires a certain amount of resource for its processing time, and can be served by any machine subject to its capacity. This differs from our model where each job has multiple tasks, each task can be served by a set of machines, and the job's completion time is determined by its last task. Minimizing the weighted sum of completion times in the PTS is also NP-complete in the strong sense~\cite{blazewicz1983scheduling}. 
In the case of a single machine, the non-preemptive algorithm in~\cite{schwiegelshohn2004preemptive} can achieve approximation ratio of $7.11$, and the preemptive algorithm in~\cite{schwiegelshohn2004preemptive}, called \textit{PSRS}, can achieve approximation ratio of $2.37$. 
 In the case of multiple machines, there is only one result in the literature which is a $14.85$-approximation non-preemptive algorithm~\cite{remy2004resource}. 

%

We emphasize that our setting of parallel-task jobs, subject to synchronization, packing, and placement constraints, is significantly more challenging than the COS and PTS problems, and 
algorithms from these problems \textit{cannot} be applied to our setting. To the best of our knowledge, this is the first paper that provides constant-approximation algorithms for this problem subject to synchronization, packing, and placement constraints,
\subsection{Main Contributions}
\label{technique}
We briefly summarize our main results and describe our techniques below. 
We propose scheduling algorithms for three cases:
\begin{itemize}[leftmargin=*]
\item \textbf{Task Migration Allowed.}
When migration is allowed, a task might be preempted several times and resume possibly on a different machine within its placement-feasible set. Our algorithm in this case is based on greedy scheduling of task fractions (fraction of processing time of each task) on each machine, subject to capacity and placement constraints. The task fractions are found by solving a relaxed linear program (LP), which divides the time horizon into geometrically-increasing time intervals, and uses \textit{interval-indexed variables} to indicate  what fraction of each task is served at which interval on each machine. 
We show that our scheduling algorithm has an approximation ratio better than $(6+\epsilon)$, for any $\epsilon>0$.

\item \textbf{Task Migration Not Allowed.}
When migration is not allowed, the schedule can be non-preemptive, or preemptive while all preemptions occur on the same machine. In this case, our algorithm is based on mapping tasks to proper time intervals on the machines. We utilize the interval-indexed variables to form a relaxed LP.
We then utilize the LP's optimal solution to construct a \textit{weighted bipartite graph} representing tasks on one side and machine-intervals on the other side, and fractions of tasks completed in machine-intervals as weighted edges. We then use an integral matching in this graph to construct a mapping of tasks to machine-intervals.
Finally, the tasks mapped to intervals of the same machine are packed in order and non-preemptively by using a greedy policy. We prove that this non-preemptive algorithm has an approximation ratio better than $24$. Further, we show that the algorithm's solution is also a 24-approximation for the case that  preemption on the same machine is allowed.

\item  \textbf{Preemption and Single-Machine Placement Set.} 
When preemption is allowed, and there is a specific machine for each task, we propose an algorithm with an improved approximation ratio of $4$.
The algorithm first finds a proper ordering of jobs, by solving a relaxed LP of our scheduling problem.
Then, for each machine, it lists its tasks, with respect to the obtained ordering of jobs, and apply a simple greedy policy to pack tasks in the machine subject to its capacity. The methods of LP relaxation and list scheduling have been used in scheduling literature; however, the application and analysis of such techniques in presence of packing, placement, and synchronization is very different.

\item \textbf{Empirical Evaluations.}  We evaluate the performance of our preemptive and non-preemptive algorithms compared with the prior approaches using a Google traffic trace~\cite{clusterdata:Wilkes2011}. We also present online versions of our algorithms that are suitable for handling dynamic job arrivals. Our $4-$approximation preemptive algorithm outperforms PSRS~\cite{schwiegelshohn2004preemptive} and Tetris~\cite{grandl2015multi} by up to $69\%$ and $79\%$, respectively, when jobs' weights are determined using their priority information in the data set. 
Further, our non-preemptive algorithm outperforms JSQ-MW~\cite{wang2016maptask} and Tetris~\cite{grandl2015multi} by up to $81\%$ and $175\%$, respectively, under the same placement constraints. 
\end{itemize}
\section{Formal Problem Statement}
\label{ProbState}
\textbf{Cluster and Job Model.} Consider a collection of machines $\calM = \{1, . . . , M\}$, where machine $i$ has capacity $m_i > 0$ on its available resource. We use $\calJ =\{1, . . . , N\}$ to denote the set of existing jobs (stages) in the system that need to be served by the machines.
Each job $j \in \calJ$ consists of a set of tasks $\calK_j$, where we use $(k,j)$ to denote task $k$ of job $j$, $k \in \calK_j$. Task $(k,j)$ requires a specific amount $a_{kj}$ of resource for the duration of its processing.
 
\textbf{Task Processing and Placement Constraint.} Each task $(k,j)$ can be processed on a machine from a specific set of machines $\calM_{kj} \subseteq \calM$. We refer to $\calM_{kj}$ as the \textit{placement set} of task $(k,j)$. For generality, we let $p_{kj}^i$  denote the processing time of task $(k,j)$ on machine $i \in \calM_{kj}$. 
Such placement constraints can model  data locality. For example, we can set $\calM_{kj}$ to be the set of machines that have task $(k,j)$'s data, and $p_{kj}^i=p_{kj}$, $i \in \calM_{kj}$. 
Or, we can consider $\calM_{kj}$ to be as large as $\calM$, and incorporate the data transfer cost as a penalty in the processing time on machines that do not have the task's data.   

Throughout the paper, we refer to $a_{kj}$ as size or resource requirement of task $(k,j)$, and to $p_{kj}^i$ as its length, duration, or processing time on machine $i$. We also define the volume of task $(k,j)$ on machine $i$ as $v_{kj}^i=a_{kj} p_{kj}^i$.
Without loss of generality, we assume processing times are nonnegative integers and duration of the smallest task is at least one. This can be done by defining a proper time unit (slot) and representing the task durations using integer multiples of this unit. 

\textbf{Synchronization Constraint.} Tasks can be processed in parallel on their corresponding machines; however, a job is considered completed only when all of its tasks finish. Hence the completion time of job $j$, denoted by $C_j$, satisfies 
\be \label{eq:jobcompletiontime}
C_j=\max_{k \in \calK_j} C_{kj},
\ee where $C_{kj}$ is the completion time of its task $(k,j)$. 

\textbf{Packing Constraint.} The sum of resource requirements of tasks running in machine $i$ should not exceed its capacity.

Let $\mathds{1}(i \in \calM_{kj})$ be $1$ if $i \in \calM_{kj}$, and $0$ otherwise. Define the constant
\begin{eqnarray}
\label{tupb}
T=\max_{i \in \calM} \sum_{j \in \calJ} \sum_{k \in \calK_j} p_{k j}^i \mathds{1}(i \in \calM_{kj}),
\end{eqnarray}
which is clearly an upper bound on the time required for processing all the jobs.
We define 0-1 variables $X_{kj}^i(t)$, $i \in \calM$, $j \in \calJ$, $k \in \calK_j$, $t\leq T$, where $X_{kj}^i(t)=1$ if task $(k,j)$ is served at time slot $t$ on machine $i$, and $0$ otherwise. 
We also make the following definition.
\begin{definition}[Height of Machine $i$ at time $t$]\label{machineheight}
The height of machine $i$ at time $t$, denoted by $h_i(t)$, is the sum of resource requirements of the tasks running at time $t$ in machine $i$, i.e.,
	\begin{eqnarray}
	h_i(t)= \sum_{j \in \calJ, k \in \calK_j} a_{kj} X_{kj}^i(t).
	\end{eqnarray}
\end{definition}
Given these definitions, a valid schedule $X_{kj}^i(t)\in \{0,1\}$, $i \in \calM$, $j \in \calJ$, $k \in \calK_j$, $0 < t \leq T$, must satisfy the following three constraints:
\begin{itemize}[leftmargin=*]
\item[(i)] \textit{Packing}: the sum of resource requirements of the tasks running in machine $i$ at time $t$ (i.e., tasks with $X_{kj}^i(t)=1$) should not exceed machine $i$'s capacity, i.e.,  $h_i(t) \leq m_i$, $\forall t \leq T$, $\forall i\in \calM$.
\item[(ii)] \textit{Placement}: each task at each time can get processed on at most a single machine selected from its feasible placement set, i.e., $\sum_{i \in \calM_{kj}}X_{kj}^i(t) \leq 1$, and $X_{kj}^i(t)=0$ if $i \notin \calM_{kj}$. 
\item[(iii)] \textit{Processing}: each task must be processed completely. Noting that ${X_{kj}^i(t)}/{p_{kj}^i}$ is the fraction of task $(k,j)$ completed on machine $i$ in time slot $t$, we need $\sum_{i \in \calM_{kj}}\sum_{t=1}^T X_{kj}^i(t)/ {p_{kj}^i}= 1$.
\end{itemize}

\textbf{Preemption and Migration.} We consider three classes of scheduling policies. In a non-preemptive policy, a task cannot be preempted (and hence cannot be migrated among machines) once it starts processing on its corresponding machine until it is completed. In a preemptive policy, a task may be preempted and resumed several times in the schedule, and we can further consider two subcases depending on whether migration of a task among machines is allowed or not. Note that when migration is not allowed, the scheduler must assign each task $(k,j)$ to one machine $i \in \calM_{kj}$ on which the task is (preemptively or non-preemptively) processed until completion.

\textbf{Main Objective.} Given positive weights $w_j$, $j \in \calJ$, our goal is to find valid non-preemptive and preemptive (under with and without migrations) schedules of jobs (their tasks) in machines, so as to minimize the sum of weighted completion times of jobs, i.e., 
\be
\mbox{minimize} \sum_{j \in \calJ} w_j C_j.
\ee
The weights can capture different priorities for jobs. Clearly the case of equal weights reduces the problem to minimization of the average completion time. 
\section{Scheduling When Migration is Allowed}
\label{preemptivemigration}
We first consider the case that migration of tasks among machines is allowed. In this case, we propose a preemptive algorithm, called {\AII}, with approximation ratio $(6+\epsilon)$ for any $\epsilon > 0$. We will use the construction ideas and analysis arguments for this algorithm to construct our preemptive and non-preemptive algorithms when migration is prohibited in Section~\ref{nomigration}.

In order to describe {\AII}, we first present a relaxed linear program. 
We will utilize the optimal solution to this LP to schedule tasks in a preemptive fashion. 
\subsection{\textbf{Relaxed Linear Program (LP1)}} Recall that without loss of generality, the processing times of tasks are assumed to be integers (multiples of a time unit) and therefore $C_j \geq p_{kj}^i \geq 1$ for all $j \in \calJ$, $k \in \calK_j$, and $i \in \calM_{kj}$. We use interval indexed variables using geometrically increasing intervals (see, e.g.,~\cite{qiu2015minimizing,queyranne20022+}) to formulate a linear program for our  problem.

Let $\epsilon>0$ be a constant. We choose $L$ to be the smallest integer such that $(1+\epsilon)^{L} \geq T$ (recall $T$ in~\dref{tupb}). Subsequently define
\be
\label{partition}
d_l= (1+\epsilon)^{l}, \mbox{ for } l=0,1,\cdots,L,
\ee
and define $d_{-1}=0$. We partition the time horizon into time intervals $(d_{l-1},d_l]$, $l=0,...,L$. 
Note that the length of the $l$-th interval, denoted by $\Delta_l$, is 
\be
\label{deltal}
\begin{aligned}
\Delta_0=1,\ \ \Delta_l=\epsilon (1+\epsilon)^{l-1} \ \ \ \forall l \geq 1.
\end{aligned}
\ee

We define $z_{kj}^{il}$ to be the fraction of task $(k,j)$ (\textit{fraction of its required processing time}) that is processed in interval $l$ on machine $i \in \calM_{kj}$. To measure completion time of job $j$, we define a variable $x_{jl}$ for each interval $l$ and job $j$ such that, $\forall j \in \calJ$:
\begin{subequations}
\begin{align}
\label{firstrel}
&\sum_{l^\prime=0}^l x_{jl^\prime} \leq {\sum_{l^\prime=0}^l \sum_{i\in \calM_{kj}} z_{kj}^{il^\prime}}, \ k \in \calK_j, \ l=0, \dots, L\\
\label{secrel}
&\sum_{l=0}^L x_{jl}=1,\ \ 
x_{jl} \in \{0,1\},  \ \  l=0, \cdots, L.
\end{align}
\end{subequations}

Note that \dref{secrel} 
implies that only one of the variables $\{x_{jl}\}_{l=0}^L$ can be nonzero (equal to $1$).  \dref{firstrel} implies that $x_{jl}$ can be $1$ only for one of the intervals  $l \geq l^\star$ where $l^\star$ is the interval in which the last task of job $j$ finishes its processing. Define, 
\be \label{eq:lowcj}
&C_j=\sum_{l=0}^L d_{l-1} x_{jl} \ \ j \in \calJ.
\ee
If we can guarantee that $x_{jl^\star}=1$ for $l^\star$ as defined above, then $C_j$ will be equal to the starting point $d_{l^\star-1}$ of that interval, and the actual completion time of job $j$ will be bounded above by $d_{l^\star}=(1+\epsilon)C_j$, thus implying that $C_j$ is a reasonable approximation for the actual completion time of job $j$. This can be done by minimizing the objective function in the following linear program (LP1):
 \begin{subequations}
	\label{RLP2}
	\begin{align}
	\label{LPobj2}
	 \min \ \ &\sum_{j \in \calJ} w_j C_j  \ \ \ \ \ \mathbf{(LP1)}\\
	\label{taskdemand}
	& \sum_{l=0}^L \sum_{i \in \calM_{kj}} z_{kj}^{il}=1, \ \ k \in \calK_j, \ j \in \calJ \\
	\label{frac}
	& \sum_{l^\prime=0}^{l}\sum_{i \in \calM_{kj}} z_{kj}^{il^\prime}p_{kj}^i \leq d_l, \ \  k \in \calK_j, \ j \in \calJ, \ l=0, \dots, L\\
	\label{intcap}
	& \sum_{l^\prime=0}^{l}\sum_{\substack{(k,j): i \in \calM_{kj}}} z_{kj}^{il^\prime} p_{kj}^i a_{kj} \leq m_i d_l, \ \ i \in \calM, \ l=0, \dots, L\\
	\label{pos}
	& z_{kj}^{il} \geq 0, \ \  k \in \calK_j, \ j \in \calJ, \ i \in \calM_{kj}, \ l=0, \dots, L\\
	\label{xdef}
	& \sum_{l^\prime=0}^l x_{jl^\prime} \leq {\sum_{l^\prime=0}^l \sum_{i\in \calM_{kj}} z_{kj}^{il^\prime}}, \ k \in \calK_j, \ j \in \calJ, \ l=0, \dots, L\\
		\label{cdef}
	& C_j=\sum_{l=0}^L d_{l-1} x_{jl}, \ \ j \in \calJ \\
	\label{jobdemand}
	& \sum_{l=0}^L x_{jl}=1,\ \ x_{jl} \geq 0,  \ \  l=0, \dots, L, \ j \in \calJ
	\end{align}
\end{subequations}

Constraint~\dref{taskdemand} means that each task must be processed completely. 
\dref{frac} is because during the first $l$ intervals, a task cannot be processed for more than $d_l$, the end point of interval $l$, which itself is due to requirement (ii) of Section~\ref{ProbState}. \dref{intcap} bounds the total volume of the tasks processed by any machine $i$ in the first $l$ intervals by $d_l \times m_i$. \dref{pos} indicates that $z$ variables have to be nonnegative. 

Constraints \dref{xdef}, \dref{jobdemand}, \dref{cdef} are the relaxed version of \dref{firstrel}, \dref{secrel}, \dref{eq:lowcj}, 
respectively, 
where the integral constraint in \dref{secrel} has been relaxed to \dref{jobdemand}. 
To give more insight, note that \dref{xdef} has the interpretation of keeping track of the fraction of the job processed by the end of each time interval, which is bounded from above by the fraction of any of its tasks processed by the end of that time interval. We should finish processing of all jobs as indicated by \dref{jobdemand}. Also \dref{cdef} computes a relaxation of the job completion time $C_j$, as a convex combination of the intervals' left points, with coefficients $x_{jl}$. Note that \dref{xdef} along with  \dref{jobdemand} implies the fact that each task is processed completely.

\subsection{\textbf{Scheduling Algorithm: \boldmath{\AII}}} \label{sec:secondpremeptive}

In the following, a \textit{task fraction} $(k,j,i,l)$ of task $(k,j)$ corresponding to interval $l$, is a task with size $a_{kj}$ and duration $z_{kj}^{il} p_{kj}^i$ that needs to be processed on machine $i$.

The {\AII} (\textit{Synchronized Packing}-$1$) algorithm has three main steps:

{\bf Step 1: Solve (LP1).} We first solve (LP1) and obtain the optimal solution of  $\{{z}_{kj}^{il}\}$ which we denote by $\{\tilde{z}_{kj}^{il}\}$.

{\bf Step 2: Pack task fractions greedily and construct schedule $\calS$.} To schedule task fractions, we use a greedy list scheduling policy as follows:

Consider an ordered list of the task fractions such that task fractions corresponding to interval $l$ appear before the task fractions corresponding to interval $l^\prime$, if $l < l^\prime$. Task fractions within each interval and corresponding to different machines are ordered arbitrarily. The algorithm scans the list starting from the first task fraction, and schedules  task fraction $(k,j,i,l)$ on machine $i$, if some fraction of task $(k,j)$ is not already scheduled on some other machine at that time, and machine $i$ has sufficient capacity, i.e., $h_i(t)+a_{kj} \leq m_i$ (recall $h_i(t)$ in Definition~\ref{machineheight}). It then moves to the next task fraction in the list, repeats the same procedure, and so on. Upon completion of a task fraction, it preempts the task fractions corresponding to higher indexed intervals on all the machines if there is some unscheduled task fraction of a lower-indexed interval in the list. It then removes the completed task fraction from the list, updates the remaining processing times of the task fractions in the list, and start scheduling the updated list. This greedy list scheduling algorithm schedules task fractions in a preemptive fashion. 

We refer to the constructed schedule as ${\calS}$.

{\bf Step 3: Apply Slow Motion and construct schedule $\bar{\calS}$.} Unfortunately, we cannot bound the value of objective function~\dref{LPobj2} for schedule $\calS$ since completion times of some jobs in ${\calS}$ can be very long compared to the completion times returned by (LP1). Therefore, we construct a new feasible schedule $\bar{\calS}$, by stretching $\calS$, for which we can bound the value of its objective function. This method is referred to as \textit{Slow-Motion} technique~\cite{schulz1997random}.
Let $\tilde{Z}_{kj}^{i}=\sum_{l=0}^{L} \tilde{z}_{kj}^{il}$ denote the total fraction of task $(k,j)$ that is scheduled in machine $i$ according to the optimal solution to (LP1). We refer to $\tilde{Z}_{kj}^{i}$ as the \emph{total task fraction} of task $(k,j)$ on machine $i$.
The Slow Motion works by choosing a parameter $\lambda \in (0,1]$ randomly drawn according to the probability density function $f(\lambda) = 2\lambda$. It then stretches schedule $\calS$ by a factor $1/\lambda$. If a task is scheduled in $\calS$ during an interval $[\tau_1, \tau_2)$, the same task is scheduled in $\bar{\calS}$ during $[\tau_1/\lambda, \tau_2/\lambda)$ and \emph{the machine is left idle if it has already processed its total task fraction $\tilde{Z}_{kj}^{i}$ completely}.
We may also shift back future tasks' schedules as far as the machine capacity allows and placement constraint is respecred.

 A pseudocode for {\AII} can be found in Appendix~\ref{pseudocode1}. The obtained algorithm is a randomized algorithm; however, we will show in Appendix~\ref{derand} how we can de-randomize it to get a deterministic algorithm.
\subsection{\textbf{Performance Guarantee}} \label{pg} 
We now analyze the performance of \AII. The result is stated by the following proposition.
\begin{theorem}
	\label{prop1}
	For any $\epsilon >0$, the sum of weighted completion times of jobs, for the problem of parallel-task job scheduling with packing and placement constraints, under \AII, is at most $(6+\epsilon) \times \textit{OPT}$.
\end{theorem}
The rest of the section is devoted to the proof of Theorem~\ref{prop1}. We use $\tilde{C}_j$ to denote the optimal solution to (LP1) for completion time of job $j \in \calJ$. The optimal objective value of (LP1) is a lower bound on the optimal value of our scheduling problem as stated in the following lemma whose proof is provided in Appendix~\ref{Proof2-1}.
\begin{lemma}
	\label{lp3bound}
	$\sum_{j=1}^N w_j \tilde{C}_j \leq \sum_{j=1}^N w_j C^\star_j = \text{OPT}$.
\end{lemma}
Note that Constraint~\dref{intcap} bounds the volume of all the task fractions corresponding to the first $l$ intervals by $d_l \times m_i$. However, the (LP1)'s solution does not directly provide a feasible schedule as task fractions of the same task on different machines might overlap during the same interval and machines' capacity constraints might be also violated. Next, we show under the greedy list scheduling policy (Step 2 in {\AII}), the completion time of task fraction $(k,j,i,l)$ is bounded from above by $3 \times d_l$, i.e., we need a factor $3$ to guarantee a feasible schedule.  
\begin{lemma}
	\label{tripledinterval}
Let $\tau_l$ denote the time that all the task fractions $(k,j,i,l^\prime)$, for $l^\prime \leq l$, are completed in schedule $\calS$.	Then, $\tau_l \leq 3d_l$.
\end{lemma}
\begin{proof}
\label{Proofvb}
Consider the non-zero task fractions $(k,j,i,l^\prime)$, $i \in \calM$, $l^\prime \leq l$ (according to an optimal solution to (LP1)). 
Without loss of generality, we normalize the processing times of task fractions to be positive integers, by defining a proper time unit  and representing the task durations using integer multiples of this unit. Let $D_l$ and $T_l$ be the value of $d_l$ and $\tau_l$ using the new unit. Let $i^\star$ denote the machine that schedules the last task fraction among the non-zero task fractions of the first $l$ intervals. Note that $T_l$ is the time that this task fraction completes. If $T_l \leq D_l$, then $T_l \leq 3 D_l$ and the lemma is proved. Hence consider the case that $T_l > D_l$.

Define $h_{il}(t)$ to be the height of machine $i$ at time $t$ in schedule $\calS$ considering only the task fractions of the first $l$ intervals. First we note that,
\begin{equation}
\label{workload}
\sum_{l^\prime=0}^{l}\sum_{{(k,j): i \in \calM_{kj}}} z_{kj}^{il^\prime} p_{kj}^i a_{kj} \stackrel{(a)}= \sum_{t=1}^{T_l} h_{il}(t) \stackrel{(b)}\leq m_i D_l,  \ \ \forall i \in \calM
 \end{equation}
Using the definition of $h_{il}(t)$, the right-hand side of Equality (a) is the total volume of task fractions corresponding to the first $l$ intervals that are processed during the interval $(0,T_l]$ on machine $i$, which is the left-hand side. Further, Inequality (b) is by Constraint~\dref{intcap}.

 Let $S_{il}(\theta)$ denote the set of task fractions, running at time $\theta$ on machine $i$. Consider machine $i^\star$. We construct a bipartite graph $G=(U \cup V, E)$\footnote{$G=(U \cup V, E)$ is a bipartite graph iff for any edge $e=(u,v) \in E$, we have $u \in U$ and $v \in V$.} as follows. For each time slot $\theta \in \{1,\dots,T_l\}$, we consider a node $z_\theta$, and define $V= \{z_\theta | 1 \leq \theta \leq T_l-D_l\}$, and $U=\{z_\theta | T_l-D_l+1 \leq \theta \leq T_l\}$. For any $z_{s} \in U $ and $z_{t} \in V$, we add an edge $(z_s,z_t)$ if $h_{i^\star l}(s)+h_{i^\star l}(t)\geq m_{i^\star}$. Note that $h_{i^\star l}(s)+h_{i^\star l}(t) < m_{i^\star}$ means, by definition, that there is no edge between $z_s$ and $z_t$. Therefore, in this case,
\begin{equation}
\label{eq:conflict}
 \big(\cup_{i \in \calM} S_{il}(s) \big) \setminus \big(\cup_{i \in \calM}S_{il}(t)\big) = \varnothing. 
\end{equation}
This is because otherwise {\AII} would have scheduled the task(s) in $S_{i^\star l}(s)$ at time $t$ (note that $t<s$).

For any set of nodes $\tilde{U} \subseteq U$, we define set of its neighbor nodes as $N_{\tilde{U}}=\{z_t \in V | \exists \ z_s \in \tilde{U}: (z_s,z_t) \in E\}$. Note that, there are $T_l-D_l-|N_{\tilde{U}}|$ nodes in $V$ which do not have any edge to some node in $\tilde{U}$. 
Let $|\cdot|$ denote set cardinality (size). We consider two cases:

\textbf{Case (i)}: There exists a set $\tilde{U}$ for which $|N_{\tilde{U}}| < |\tilde{U}|$. Consider a node $z_s \in \tilde{U}$ and a task with duration $p$ running at time slot $s$. Let $p_U$ denote the amount of time that this task is running on time slots of set $U$. Note that $p_U\geq1$. By Equation~\dref{eq:conflict}, a task that is running at time $s$ is also running at $T_l-D_l-|N_{\tilde{U}}|$ many other time slots whose  corresponding nodes are in $V$. 
\begin{equation*}
p=T_l-D_l-|N_{\tilde{U}}|+p_U \leq D_l,
\end{equation*}
where the inequality is by Constraint~\dref{frac}. Therefore
\begin{equation*}
T_l \leq 2D_l+|N_{\tilde{U}}|-p_U <2D_l+|\tilde{U}| \leq 3D_l.
\end{equation*}

\textbf{Case (ii)}: For any $\tilde{U} \subseteq U$,  $|\tilde{U}| \leq |N_{\tilde{U}}|$. Hence, $|V| \geq |U|$ which implies that $T_l\geq 2D_l$. Further, Hall's Theorem~\cite{hall1935representatives} states that a perfect matching of nodes in $U$ to nodes in $V$ always exists in $G$\footnote{A perfect matching in $G$ (with size $|U|$) is a subset of $E$ such that every node in set $U$ is matched to one and only one node in set $V$ by an edge in the subset.} in this case. The existence of such a matching then implies that any time slot $s \in (T_l-D_l,T_l]$ can be matched to a time slot $t_s \in (0,T_l-D_l]$ and $h_{i^\star l}(s)+h_{i^\star l}(t_s) \geq m_i$. This implies that 
\begin{equation}
\sum_{s \in U }(h_{i^\star l}(s)+h_{i^\star l}(t_s)) \geq m_{i^\star} D_l \stackrel{(c)}\geq \sum_{t=1}^{T_l}h_{i^\star l}(t),
\end{equation}
where Inequality (c) is by Equation~\dref{workload}. From this, one can conclude that no non-zero task fraction $(k,j,i^\star,l^\prime)$, $i^\star$, $l^\prime \leq l$ is processed at time slots $V^\prime=V \setminus \cup_{s \in U } \{t_s\}$. Hence, $V^\prime=\varnothing$, since otherwise {\AII} would have scheduled some of the tasks running at time slots of set $U$, at $V^\prime$. We then can conclude that $T_l=2D_l<3D_l$. This completes the proof.
\end{proof}

Next, we make the following definition regarding schedule $\calS$.
\begin{definition}
	\label{calpha}
	We define $C_j(\alpha)$, for $0 < \alpha \leq 1$, to be the time at which $\alpha$-fraction of job $j$ is completed in schedule $\calS$ (i.e., at least $\alpha$-fraction of each of its tasks has been completed.).
\end{definition}
The following lemma shows the relationship between $C_j(\alpha)$  and $\tilde{C}_j$, the optimal solution to (LP1) for completion time of job $j$. The proof is provided in Appendix~\ref{Proof2-5}.
\begin{lemma}
	\label{intcjalpha}
	$\int_{\alpha=0}^{1} C_j(\alpha)d\alpha \leq 3(1+\epsilon)\tilde{C}_j$
\end{lemma}
Recall that schedule $\bar{\calS}$ is formed by stretching schedule $\calS$ by factor $1/\lambda$. Let $\bar{C}_j^\lambda$ denote the completion time of job $j$ in $\bar{\calS}$. Then we can show that the following lemma holds.
\begin{lemma}
	\label{smbound}
	$\expect{\bar{C}_j^\lambda} \leq 6(1+\epsilon) \tilde{C_j}$. 
\end{lemma}
\begin{proof}
	The proof is based on Lemma~\ref{intcjalpha} and taking expectation with respect to probability density function of $\lambda$. The details can be found in Appendix~\ref{Proof2-6}.
\end{proof}
In constructing $\bar{\calS}$, we may shift scheduling time of some of the tasks on each machine to the left and construct a better schedule. Nevertheless, we have the performance guarantee of Theorem~\ref{prop1} even without this shifting.

\begin{proof}[Proof of Theorem~\ref{prop1}]
	Let $C_j$ denote the completion time of job $j$ under {\AII}. Then
	\begin{equation*}
	\begin{aligned}
	\expect{\sum_{j \in \calJ} w_j C_j} & \leq \expect{\sum_{j \in \calJ} w_j \bar{C}_j^\lambda} \stackrel{(a)}\leq 6(1+\epsilon)\sum_{j \in \calJ} w_j\tilde{C_j}  \\ &\stackrel{(b)}\leq 6(1+\epsilon) \sum_{j \in \calJ} w_j C_j^\star,
	\end{aligned}
	\end{equation*}
	where (a) is by Lemma~\ref{smbound}, and (b) is by Lemma~\ref{lp3bound}. In Appendix~\ref{derand}, we discuss how to de-randomize the random choice of $\lambda \in (0, 1]$, which is used to construct schedule $\bar{\calS}$ from schedule $\calS$. So the proof is complete.
\end{proof}

\section{Scheduling When Migration is not Allowed}
\label{nomigration}
The algorithm in Section~\ref{preemptivemigration} is preemptive, and tasks can be migrated across the machines in the same placement set. Implementing such an algorithm can be complex and costly in practice. In this section, we consider the case that migration of tasks among machines is not allowed. We propose a non-preemptive scheduling algorithm for this case. We also show that its solution provides a bounded solution for the case that preemption of tasks (in the same machine, without migration) is allowed. 

Our algorithm is based on a relaxed LP 
which is very similar to (LP1) of Section~\ref{preemptivemigration}, however a different constraint is used to ensure that each task is scheduled entirely by the end point of some time-interval of a machine. Next, we introduce this LP and describe how to generate a non-preemptive schedule based on its solution. 

\subsection{Relaxed Linear Program (LP2)} 
We partition the time horizon into intervals $(d_{l-1},d_l]$ for $l=0,...,L$, as defined in \dref{partition} by replacing $\epsilon$ by $1$. Define 0-1 variable $z_{kj}^{il}$ to indicate whether task $(k,j)$ is completed on machine $i$ by the end-point of interval $l$, i.e., by $d_l$. Note that the interpretation of variables $z_{kj}^{il}$ is slightly different from their counterparts in (LP1).
By relaxing integrality of $z$ variables, we formulate the following LP:
\begin{subequations}
	\label{RLP3}
	\begin{align}
	\label{LPobj3}
	 \min \ \ &\sum_{j \in \calJ} w_j C_j\ \ \quad \mathbf{(LP2)} \\
		 \label{extraconst}
		 & z_{kj}^{il}=0\  \text{ if  } p_{kj}^i > d_l,\ \  j \in \calJ, \ k \in \calK_j, \ i \in \calM_{kj}, \ l=0, \dots, L \\
		 & \text{Constraints~\dref{taskdemand}--\dref{jobdemand}}		 
    \end{align}
\end{subequations}
Note that Constraint~\dref{extraconst} allows $z_{kj}^{il}$ to be positive only if the end point of $l$-th interval is at least as long as task $(k,j)$'s processing time on machine $i \in \calM_{kj}$. We would like to emphasize that this is a valid constraint for both the preemptive and non-preemptive cases when migration is not allowed. We will see shortly how this constraint helps us construct our non-preemptive algorithm. We interpret \emph{fractional} values of $z_{kj}^{il}$ as the fraction of task $(k,j)$ that is processed in interval $l$ of machine $i$ (as in Section~\ref{preemptivemigration}).  

\subsection{Scheduling Algorithm: \boldmath{{\AIII}}}
\label{sec:nonpremeptive}  
Our non-preemptive algorithm, which we refer to as {\AIII}, has three main steps:

{\bf Step 1: Solve (LP2).} We first solve the linear program (LP2) to obtain the optimal solution of $\{{z}_{kj}^{il}\}$ denoted by $\{\tilde{z}_{kj}^{il}\}$.

{\bf Step 2: Apply Slow-Motion.}
Before constructing the actual schedule of tasks, the algorithm applies the Slow-Motion technique (see Section~\ref{sec:secondpremeptive}). 
We pause here to clarify the connection between $\tilde{z}_{kj}^{il}$ and those obtained after applying Slow-Motion which we denote by $\bar{z}_{kj}^{il}$, below. 

Recall that $\tilde{z}_{kj}^{il}$ is the fraction of task $(k,j)$ that is scheduled in interval $l$ of machine $i$ in the optimal solution to (LP2), and $\Delta_l$ is the length of the $l$-th interval. Also, recall that $\tilde{Z}_{kj}^{i}=\sum_{l=0}^{L} \tilde{z}_{kj}^{il}$ is the total task fraction to be scheduled on machine $i$ corresponding to task $(k,j)$.
Similarly, we define $\bar{\Delta}_l$ and $\bar{d}_l$ to be the length and the end point of the $l$-th interval after applying the Slow-Motion using a stretch parameter $\lambda \in (0,1]$, respectively. Therefore, 
\begin{equation}
\label{barintlength}
\bar{\Delta}_l=\frac{\Delta_l}{\lambda}, \ \ \bar{d}_l=\frac{d_l}{\lambda} .
\end{equation} 
Further, we define $\bar{z}_{kj}^{il}$ to be the fraction of task $(k,j)$ to be scheduled during the $l$-th interval on machine $i$ after applying Slow-Motion. Then it holds that,
\begin{equation}
\label{zbardef}
\bar{z}_{kj}^{il}=
\begin{cases}
\frac{\tilde{z}_{kj}^{il}}{\lambda}, & \text{if}\ \sum_{l\prime=0}^{l} \frac{\tilde{z}_{kj}^{il^\prime}}{\lambda} < \tilde{Z}_{kj}^{i}\\
\max \left\{0,\Big(\tilde{Z}_{kj}^{i}- \sum_{l\prime=0}^{l-1}\frac{\tilde{z}_{kj}^{il^\prime}}{\lambda} \Big)\right\},
& \text{otherwise.}
\end{cases}
\end{equation}
To see \dref{zbardef}, note that in Slow-Motion, both variables and intervals are stretched by factor $1/\lambda$, and after stretching, the machine is left idle if it has already processed its total task fraction completely. Hence, as long as $\tilde{Z}_{kj}^{i}$ fraction of task $(k,j)$ is not completely processed by the end of the $l$-th interval in the stretched solution, it is processed for $\tilde{z}_{kj}^{il}p_{kj}^i/\lambda$ amount of time in the $l$-th interval of length $\bar{\Delta}_l=\Delta_l/\lambda$. Hence $\bar{z}_{kj}^{il}=\tilde{z}_{kj}^{il}/\lambda$. Now suppose  $l^\star$ is the first interval for which $\sum_{l\prime=0}^{l^\star} \tilde{z}_{kj}^{il^\prime}/\lambda \geq \tilde{Z}_{kj}^{i}$. Then, the remaining processing time of task $(k,j)$ to be scheduled in the $l^\star$-th interval of machine $i$ in the stretched schedule 
is $p_{kj}^i(\tilde{Z}_{kj}^{i}- \sum_{l\prime=0}^{l^\star-1}\bar{z}_{kj}^{il^\prime})=
p_{kj}^i(\tilde{Z}_{kj}^{i}- \sum_{l\prime=0}^{l^\star-1}\tilde{z}_{kj}^{il^\prime}/\lambda) >0
$. Therefore, the second part of \dref{zbardef} holds for $l^\star$, and for intervals $l >l^\star$, $\bar{z}_{kj}^{il}$ will be zero, since $\tilde{Z}_{kj}^{i}- \sum_{l\prime=0}^{l-1}\bar{z}_{kj}^{il^\prime}/\lambda \leq 0$. Observe that $\sum_{i \in \calM_{kj}}\sum_{l=0}^L \bar{z}_{kj}^{il}=1$.

{\bf Step 3: Construct a non-preemptive schedule.} Note that according to variables $\bar{z}_{kj}^{il}$, a task possibly is set to get processed in different intervals and machines. The last step of {\AIII} is the procedure of constructing a non-preemptive schedule using these variables. This procedure involves 2 substeps: (1) \textit{mapping of tasks to machine-intervals}, and (2) \textit{non-preemptive scheduling of tasks mapped to each machine-interval using a greedy scheme}. We now describe each of these substeps in detail.

\textbf{Substep 3.1: Mapping of tasks to machine-intervals.} For each task $(k,j)$, the algorithm uses a mapping procedure to find \textit{a machine and an interval in which it can schedule the task entirely in a non-preemptive fashion}. 
The mapping procedure is based on constructing a weighted bipartite graph $\calG=(U \cup V, E)$, followed by an \textit{integral matching} of nodes in $U$ to nodes in $V$ on edges with non-zero weights, as described below:
\begin{itemize}[leftmargin=*]
\item [(i)] \textit{Construction of Graph $\calG=(U \cup V, E)$:} For each task $(k,j)$, $j \in \calJ$ , $k \in \calK_j$, we consider a node in $U$. Therefore, there are $\sum_{j \in \calJ}|\calK_j|$ nodes in $U$. Further, $V=\cup_{i \in \calM} V_i$, where $V_i$ is the set of nodes that we add for machine $i$. To construct graph $\calG$, we start from the first machine, say machine $i$, and sort tasks in non-increasing order of their volume $v_{kj}^i=a_{kj}p_{kj}^i$ in machine $i$. 
Let $N_i$ denote the number of tasks on machine $i$ with nonzero volumes. Without loss of generality, suppose
\be
\label{order3}
v_{k_1j_1}^i \geq v_{k_2j_2}^i \geq \dots v_{k_{N_i}j_{N_i}}^i >0.
\ee
For each interval $l$, we consider $\lceil \bar{z}^{il} \rceil=\lceil \sum_{j \in J} \sum_{k \in \calK_j} \bar{z}_{kj}^{il} \rceil$ (recall the definition of $\bar{z}_{kj}^{il}$ in \dref{zbardef}) consecutive nodes in $V_i$ which we call \textit{copies of interval $l$}. 

Starting from the first task in the ordering~\dref{order3}, we draw edges from its corresponding node in $U$ to the interval copies in $V_i$ in the following manner. Assume we reach at task $(k,j)$ in the process of adding edges. For each interval $l$, if $\bar{z}_{kj}^{il} >0$, first set $R=\bar{z}_{kj}^{il}$. Consider the first copy of interval $l$ for which the total weight of its current edges is strictly less than $1$ and set $W$ to be its total weight. We draw an edge from the node of task $(k,j)$ in $U$ to this copy node in $V_i$, and assign a weight equal to $\min\{R,1-W\}$ to this edge. Then we update $R \leftarrow R-\min\{R,1-W\}$, consider the next copy of interval $l$, and apply the same procedure, until $R=0$ (or equivalently, the sum of edge weights from node $(k,j)$ to copies of interval $l$ becomes equal to $\bar{z}_{kj}^{il}$). We use $w_{kj}^{ilc}$ to denote the weight of edge that connects task $(k,j)$ to copy $c$ of interval $l$ of machine $i$, and if there is no such edge, $w_{kj}^{ilc}=0$.
We then move to the next machine and apply the similar procedure, and so on. See Figure~\ref{exm} for an illustrative example.
\begin{figure}[t]
	\centering
	\includegraphics[scale=0.25]{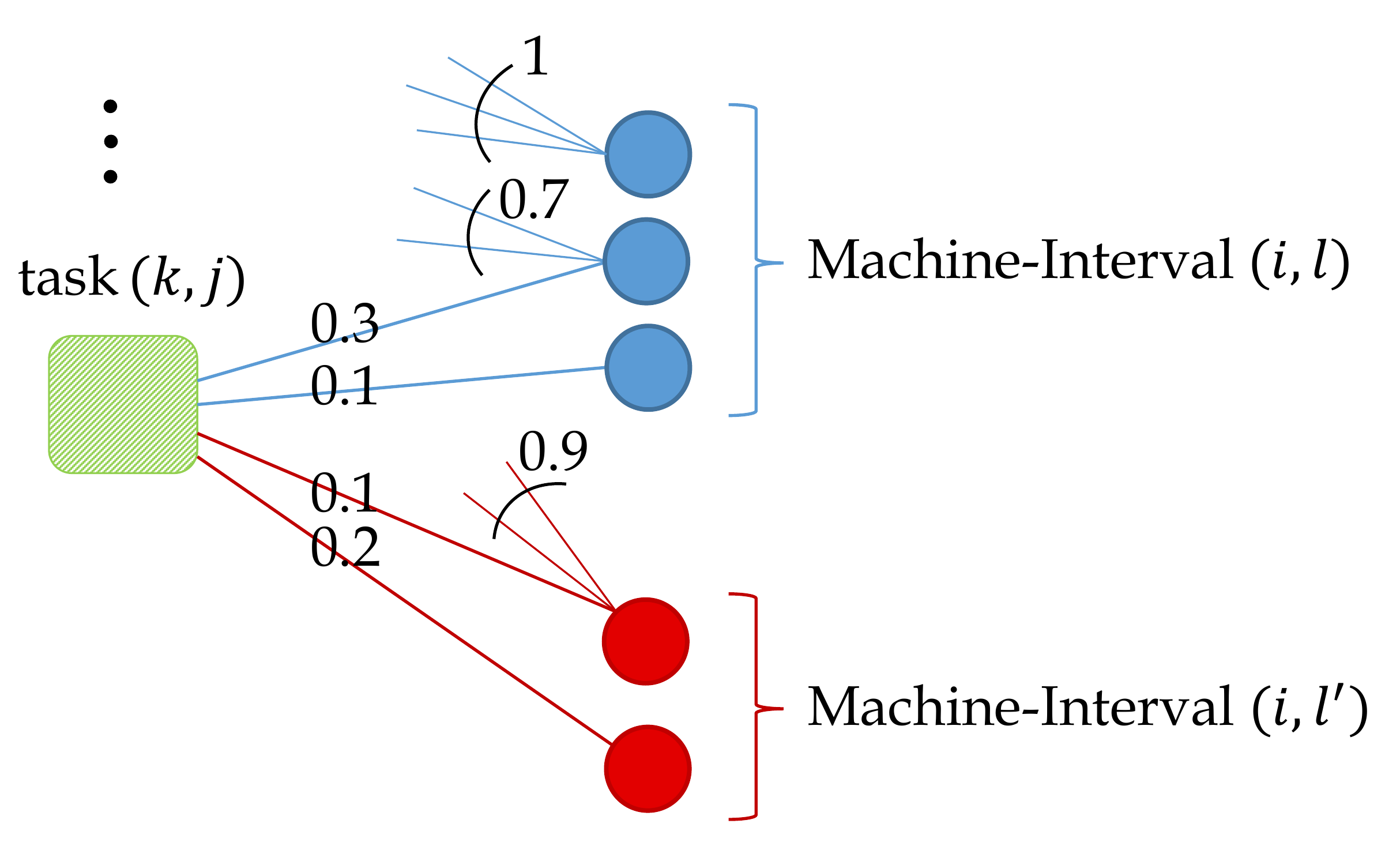}%
	\caption{An illustrative example for construction of graph $\calG$ in Substeb 3.1. Task $(k,j)$ requires $\bar{z}_{kj}^{il}=0.4$ and $\bar{z}_{kj}^{il^\prime}=0.3$. When we reach at task $(k,j)$, the total weight of the first copy of interval $l$ is $1$ and that of its second copy is $0.7$. Also, the total weight of the first copy of interval $l^\prime$ is $0.9$. Hence, the procedure adds $2$ edges to copies of interval $l$ with weights $0.3$ and $0.1$, and $2$ edges to copies of interval $l^\prime$ with weights $0.1$ and $0.2$. }
	\label{exm}%
	\vspace{-12pt}
\end{figure}

Note that in $\calG$, the weight of any node $u \in U$ (the sum of weights of its edges) is equal to $1$ (since $\sum_{l=0}^L \bar{z}_{kj}^{il}=1$, for any task $(k,j)$), while the weight of any node $v \in V$ is at most $1$.  

\item[(ii)] \textit{Integral Matching:} Finally, we find an integral matching on the non-zero edges of $\calG$, such that each non-zero task is matched to some interval copy. As we will show shortly in Section~\ref{sec:performance nonpreemptive}, we can always find an integral matching of size $\sum_{j \in \calJ}|\calK_j|$, 
the total number of tasks, in $\calG$, in polynomial time, in which each task is matched to a copy of some interval.
\end{itemize}
A pseudocode for the mapping procedure can be found in Appendix~\ref{pseudocode3}.

\textbf{Substep 3.2: Greedy packing of tasks in machine-intervals.}
We utilize a greedy packing to schedule all the tasks that are mapped to a machine-interval \textit{non-preemptively}. More precisely, on each machine, the greedy algorithm starts from the first interval and considers an arbitrary ordered list of its corresponding tasks. Starting from the first task, the algorithm schedules it, and moves to the second
task. If the machine has sufficient capacity, it schedules the task, otherwise it checks the next task and so on. Once it is done with all the tasks of the first interval, it considers the second interval, applies the similar procedure, and so on. We may also shift back future tasks' schedules as far as the machine capacity allows.

Note that this greedy algorithm is simpler than the one described in Section~\ref{preemptivemigration}, since it does not need to consider requirement (ii) of Section~\ref{ProbState} as here each task only appears in one feasible machine.

As we prove in the next section, we can bound the total volume of tasks mapped to interval $l$ on machine $i$ in the mapping phase by $m_i \bar{\Delta}_l$. Furthermore, by Constraint~\dref{extraconst} and the fact that the integral matching in Substep 3.1 was constructed on non-zero edges, the processing time of any task mapped to an interval is not greater than the interval's end point, which is twice the interval length. Hence, we can bound the completion time of each job and find the approximation ratio that our algorithm provides.

A pseudocode for the {\AIII} algorithm can be found in Appendix~\ref{pseudocode3}.

\subsection{Performance Guarantee} \label{sec:performance nonpreemptive}
In this section, we analyze the performance of our non-preemptive algorithm \AIII. The main result of this section is as follows:
\begin{theorem}
	\label{theorem2}
	The scheduling algorithm \AIII, in Section~\ref{sec:nonpremeptive}, is a $24$-approximation algorithm for the problem of parallel-task jobs scheduling with packing and placement constraints, when preemption and migration is not allowed.
\end{theorem}
Since the constraints of (LP2) also hold for the preemptive case when migration is not allowed, the optimal solution of this case is also lower bounded by the optimal solution to the LP. Therefore, the algorithm' solution is also a bounded solution for the case that preemption is allowed (while still migration is not allowed).
\begin{corollary}
	The scheduling algorithm \AIII, in Section~\ref{sec:nonpremeptive}, is a $24$-approximation algorithm for the problem of parallel-task jobs scheduling with packing and placement constraints, when preemption is allowed and migration is not.
\end{corollary} 
The rest of this section is devoted to the proof of Theorem~\ref{theorem2}. With a minor abuse of notation, we use $\tilde{C}_{kj}$ and $\tilde{C}_j$ to denote the completion time of task $(k,j)$ and job $j$, respectively, in the optimal solution to (LP2). Also, let $C_{kj}^\star$ and $C_j^\star$ denote the completion time of task $(k,j)$ and job $j$, respectively, in the optimal non-preemptive schedule. We can bound the optimal value of (LP2) as stated below. The proof is provided in Appendix~\ref{Proof3-1}.
\begin{lemma}
	\label{lp4bound}
	$\sum_{j=1}^N w_j \tilde{C}_j \leq \sum_{j=1}^N w_j C^\star_j = OPT$.
\end{lemma} 
\begin{definition}\label{def:counterpart}
Given $0 < \alpha \leq 1$, define $\hat{C}_j(\alpha)$ to be the starting point of the earliest interval $l$ for which $\alpha \leq \tilde{x}_{jl}$ in solution to (LP2).
\end{definition} 
Note that $\hat{C}_j(\alpha)$ is slightly different from Definition~\ref{calpha}, as we do not construct an actual schedule yet. We then have the following corollary which is a counterpart of Lemma~\ref{intcjalpha}. See Appendix~\ref{Proof3-4} for the proof.
\begin{corollary}
	\label{coro}
	$\int_{\alpha=0}^{1} \hat{C}_j(\alpha)d\alpha =\tilde{C}_j$
\end{corollary}
Consider the mapping procedure where we construct bipartite graph $\calG$ and match each task to a copy of some machine-interval. Below, we state a lemma which ensures that indeed we can find an integral (i.e. $0$ or $1$) matching in $\calG$. The proof can be found in Appendix~\ref{Proof3-2}.
\begin{lemma}
	\label{intmatching}
	Consider graph $\calG$ constructed in the mapping procedure. There exists an integral matching on the nonzero edges of $\calG$ in which each task is matched to some interval copy. Further, this matching can be found in polynomial time.
\end{lemma}
Let $V_{il}$ denote the total volume of the tasks mapped to all the copies of interval $l$ of machine $i$. The following lemma bounds $V_{il}$ whose proof is provided in Appendix~\ref{Proof3-3}.
\begin{lemma}
	\label{machineintervalvol}
	For any machine-interval $(i,l)$, we have 
	\be \label{eq:machineintervalvol}
	V_{il} \leq \bar{d}_l m_i + \sum_{j \in \calJ} \sum_{k \in \calK_j} v_{kj}^i \bar{z}_{kj}^{il}.
	\ee
\end{lemma}
Note that the second term in the right side of \dref{eq:machineintervalvol} can be bounded by $\bar{d}_l m_i$ which results in the inequality $V_{il} \leq 2\bar{d}_l m_i$. However, the provided bound is tighter and allows us to prove a better bound for the algorithm.
We next show that, using the greedy packing algorithm, we can schedule all the tasks of an interval $l$ in a bounded time.
 
In the case of packing single tasks in a single machine, the greedy algorithm is known to provide a $2$-approximation solution for minimizing  makespan~\cite{garey1975bounds}. 
The situation is slightly different in our setting as we require to bound the completion time of the last task  as a function of the total volume of tasks, when the maximum duration of all tasks in each interval is bounded. We state the following lemma and its proof in Appendix~\ref{Proof2-2} for completeness. 
\begin{lemma}
	\label{doubledinterval}
	Consider a machine with capacity $1$ and a set of tasks $J=\{1,2,\dots,n\}$. Suppose each task $j$ has size $a_j \leq 1$, processing time $p_j \leq 1$, and $\sum_{j \in J} a_j p_j \leq v$. Then, we can schedule all the tasks within the interval $(0,2\max\{1,v\}]$ using the greedy algorithm.
\end{lemma}
Now consider a machine-interval $(i,l)$. Note that Lemma~\ref{machineintervalvol} bounds the total volume of tasks while Constraint~\dref{extraconst} ensures that duration of each task is less than $d_l$. Thus, by applying Lemma~\ref{doubledinterval} on the normalized instance, in which size and length of tasks are normalized by $m_i$ and $d_l$, respectively, we guarantee that we can schedule all the task within a time interval of length $2d_l+2\sum_{j \in \calJ} \sum_{k \in \calK_j} v_{kj}^i \bar{z}_{kj}^{il}/m_i$. Moreover, the factor $2$ is tight as stated in the following lemma. The proof can be found in Appendix~\ref{Proof2-3}.
\begin{lemma}
	\label{tight}
We need an interval of length at least $2\max(1,v)$ to be able to schedule any list of tasks as in Lemma~\ref{doubledinterval} using any algorithm.
\end{lemma}
Hence, Lemmas~\ref{doubledinterval} and \ref{tight} imply that applying the greedy algorithm to schedule the tasks of each machine-interval, provides a tight bound with respect to the total volume of tasks in that machine-interval.
Let $C_{kj}$ denote the completion time of task $(k,j)$ under \AIII. Then we have the following lemma.
\begin{lemma}
	\label{ctbound}
	Suppose that task $(k,j)$ is mapped to the $l$-th interval of machine $i$ at the end of Substep 3.1. Then, $C_{kj}\leq 6 \bar{d}_l$.
\end{lemma}
\begin{proof}
	Let $T_{il}$ denote the completion time of the last task of machine-interval $(i,l)$, and   $\tau_{il^\prime}$ be the length of the time interval that {\AIII} uses to schedule tasks of machine-interval $(i,l)$. Then,
	\begin{equation}
	\begin{aligned}
	C_{kj} &\leq T_{il}=\sum_{l^\prime=0}^l \tau_{il^\prime} \stackrel{(a)} \leq 2 \times \sum_{l^\prime=0}^l \big (\bar{d}_{l^\prime}+ \sum_{j^\prime \in \calJ} \sum_{k^\prime \in \calK_{j^\prime}} v_{k^\prime j^\prime}^i \bar{z}_{k^\prime j^\prime}^{il^\prime}/m_i\big )\\
	& \stackrel{(b)}\leq 4 \bar{d}_{l}+2\sum_{l^\prime=0}^l \sum_{j^\prime \in \calJ} \sum_{k^\prime \in \calK_{j^\prime}} v_{k^\prime j^\prime}^i \bar{z}_{k^\prime j^\prime}^{il^\prime}/m_i\stackrel{(c)}\leq 6\bar{d}_{l}.
	\end{aligned}
	\end{equation}
	Inequality (a) is due to Lemma~\ref{machineintervalvol} and Lemma~\ref{doubledinterval}, while Inequality (b) is because $d_{l^\prime-1}=d_{l^\prime}/2$. Further, Inequality (c) is by Constraint~\dref{intcap}.
\end{proof}
\begin{proof}[Proof of Theorem~\ref{theorem2}]
	Let $l$ denote the end point of the interval in which task $(k,j)$ has the last non-zero fraction according to $\bar{z}_{kj}^{il}$. Then,
	\begin{equation}
	\label{helpineq}
	\bar{d}_l=2^l/\lambda \stackrel{(\star)}\leq 2\hat{C}_j(\lambda)/\lambda.
	\end{equation}
	  First note that $\epsilon$ is replaced by $1$ in Equation~\dref{partition}. Further, Inequality $(\star)$ follows from the definition of $\hat{C}_j(\lambda)$ (Definition~\ref{def:counterpart}), and the fact that $d_l$'s are multiplied by $1/\lambda$. Therefore, $\hat{C}_j(\lambda)/\lambda$ is the start point of the interval in which job $j$ is completed, and, accordingly, $2\hat{C}_j(\lambda)/\lambda$ is the end point of that interval. Thus, $2^l/\lambda $, the end point of the interval in which task $(k,j)$ is completed, has to be at most $2\hat{C}_j(\lambda)/\lambda$, the end point of the interval in which job $j$ is completed. 
	 
	 Let $C_{kj}$ and $C_j$ be the completion time of task $(k,j)$ and job $j$ under \AIII. Recall that in the mapping procedure, we only map a task to some interval $l^\prime$ in which part of the task is assigned to that interval after Slow-Motion applied (in other words, $\bar{z}_{kj}^{il^\prime} > 0$). Thus, task $(k,j)$ that has its last non-zero fraction in interval $l$ (by our assumption) is mapped to some interval $l^\prime \leq l$, because $\bar{z}_{kj}^{il''}=0$ for intervals $l'' > l$. Suppose task $(k_j,j)$ is the last task of job $j$ and finishes in interval $l_j$ in our non-preemptive schedule. Then, by Lemma~\ref{ctbound} and Equation~\dref{barintlength}, we have $C_j=C_{i_jj} \leq 6 \bar{d}_l= \frac{6}{\lambda} 2^{l_j}$.
Recall that $\tilde{C}_j$ denotes the completion time of job $j$ in an optimal solution of (LP2). Hence, 
	\begin{equation*}
	\begin{aligned}
	\expect{\sum_{j \in \calJ} w_j C_j} &\leq \expect{ \sum_{j \in \calJ} w_j\frac{6}{\lambda} 2^{l_j}} \stackrel{(a)} \leq 12 \times \expect{ \sum_{j \in \calJ} w_j \hat{C}_j(\lambda)/\lambda}\\
	& \stackrel{(b)}= 12 \times \sum_{j \in \calJ} w_j \int_{\lambda=0}^{1} \frac{\hat{C}_j(\lambda)}{\lambda} 2\lambda d\lambda \stackrel{(c)} \leq 24 \times \sum_{j \in \calJ} w_j\tilde{C}_j,
	\end{aligned}
	\end{equation*}
	where in the above, (a) is by the second part of \dref{helpineq} for $l=l_j$, (b) is by definition of expectation with respect to $\lambda$, with pdf $f(\lambda)=2\lambda$, and
	 (c) is by Corollary~\ref{coro}. Using the above inequality and Lemma~\ref{lp4bound},
	\begin{equation}
	\label{final}
	\begin{aligned}
		\expect{\sum_{j \in \calJ} w_j C_j} \leq 24 \times \sum_{j \in \calJ} w_j C_j^\star=24 \times \text{OPT}.
	\end{aligned}
	\end{equation}
	
	By applying de-randomization procedure (see Appendix~\ref{derand}), we can find $\lambda=\lambda^\star$ in polynomial time for which the total weighted completion time is less that its expected value in \dref{final}. This completes the proof of Theorem~\ref{theorem2}.
\end{proof}
\section{Special Case: Preemption and Single-Machine Placement set}
\label{specialcase}
In previous sections, we studied the parallel-task job scheduling problem for both cases when migration of tasks (among machines in its placement set) is allowed or not, and provided $(6+\epsilon)$ and $24$ approximation algorithms, respectively. In this section, we consider a special case when only one machine is in the placement set of each task (e.g., it is the only machine that has the required data for processing the task), and preemption is allowed.
\begin{corollary}
	Consider the parallel-task job scheduling problem when there is a specific machine to process each task and preemption is allowed. For any $\epsilon >0$, the sum of the weighted completion times of jobs under \AII, in Section~\ref{sec:secondpremeptive}, is at most $(4+\epsilon) \times \textit{OPT}$.
\end{corollary}
\begin{proof}
	The proof is straight forward and similar to proof of Theorem~\ref{prop1}. Specifically, the factor $3$ needed to bound the solution of the greedy policy is reduced to $2$ due to the fact that placement constraint is not needed to be enforced here, since there is only one machine for each task.
\end{proof}
We can show that there is a slightly better approximation algorithm to solve the problem in this special case, that has an approximation ratio $4$. The algorithm 
uses a relaxed LP, based on linear ordering variables (e.g., \cite{gandhi2008improved,mastrolilli2010minimizing,shafiee2017improved}) to find an efficient ordering of jobs. Then it applies a simple list scheduling to pack their tasks in machines subject to capacity constraints.  The details are as follows.  
\subsection{\textbf{Relaxed Linear Program (LP3)}}
Note that each task has to be processed in a specific machine. Each job consists of up to $M$ (number of machines) different tasks. We use $\calM_j$ to denote the set of machines that have tasks for job $j$. Task $i$ of job $j$, denoted as task $(i,j)$, requires a specific amount $a_{ij}$ of machine $i$’s resource ($a_{ij} \leq m_i$) for a specific time duration $p_{ij}>0$. We also define its volume as $v_{ij}=a_{ij}p_{ij}$. The results also hold in the case that a job has multiple tasks on the same machine. 

For each pair of jobs, we define a binary variable which indicates which job is finished before the other one. Specifically, we define $\delta_{j j^\prime} \in \{0,1\}$ such that $\delta_{j j^\prime}=1$ if job $j$ is completed before job $j^\prime$, and $\delta_{j j^\prime}=0$ otherwise. Note that by the synchronization constraint \dref{eq:jobcompletiontime}, the completion of a job is determined by its last task. If both jobs finish at the same time, we set either one of $\delta_{j j^\prime}$ or $\delta_{j^\prime j}$ to $1$ and the other one to $0$, arbitrarily. By relaxing the integral constraint on binary variables, we formulate the following LP:
\begin{subequations}
	\label{RLP1}
	\begin{align}
	\label{LPobj1}
	\min \ \ &\sum_{j \in \calJ} w_j C_j\ \ \quad \mathbf{(LP3)} \\
	\label{volume}
	& m_i C_j \geq  v_{ij}+ \sum_{j^\prime \in \calJ, j^\prime \neq j } v_{ij^\prime} \delta_{j^\prime j}, \ \  j \in \calJ, i \in \calM_j\\
	\label{process}
	& C_j \geq p_{ij}, \ \ j \in \calJ, i \in \calM_j\\
	\label{prec1}
	&\delta_{j j^\prime}+\delta_{j^\prime j}=1, \ \  j \neq j^\prime, \ j,j^\prime \in \calJ\\
	\label{pstv}
	&\delta_{j j^\prime} \geq 0, \ \ j,j^\prime \in \calJ
	\end{align}
\end{subequations}
Recall the definition of job completion time $C_j$ and task completion time $C_{ij}$ in Section~\ref{ProbState}. In (LP3), \dref{volume} follows from the definition of  $\delta_{j j^\prime}$, and the fact that the tasks which need to be served on machine $i$ are processed by a single machine of capacity $m_i$. 
It states that the total volume of tasks that can be processed during the time period $(0,C_j]$ by machine $i$ is \textit{at most} $m_iC_j$. This total volume is given by the right-hand-side of~\dref{volume} which basically sums the volumes of the tasks on machine $i$ that finish before job $j$ finishes its corresponding tasks at time $C_j$, plus the volume of task $(i,j)$ itself. 
Constraint~\dref{process} is due to the fact that $C_j \geq C_{ij}$ and each task cannot be completed before its processing time $p_{ij}$. \dref{prec1} indicates that for each two jobs, one precedes the other. Further, we relax the binary ordering variables to be fractional in \dref{pstv}.   

Note that the optimal solution to (LP3) might be an infeasible schedule as (LP3) replaces the tasks by sizes of their volumes and it might be impossible to pack the tasks in a way that matches the obtained completion times from (LP3).
\begin{remark}
	(LP3) can be easily modified to allow each job to have multiple tasks on the same machine. We omit the details to focus on the main ideas. 
\end{remark}  
\subsection{\textbf{Scheduling Algorithm: \boldmath{\AI}}}~\\
\label{sec:firstpremeptive}
The {\AI} algorithm has two steps: 

{\bf Step 1: Solve (LP3) to find an ordering of jobs.} Let $\tilde{C}_j$ denote the optimal solution to (LP3) for completion time of job $j \in \calJ$. We order jobs based on their $\tilde{C}_j$ values in a nondecreasing order. Without loss of generality, we re-index the jobs such that
\begin{equation}
\label{order}
\tilde{C}_1 \leq \tilde{C}_2 \leq ... \leq \tilde{C}_N.
\end{equation}
Ties are broken arbitrarily. 

{\bf Step 2: List scheduling based on the obtained ordering.}
For each machine $i$, the algorithm maintains a list of tasks such that for every two tasks $(i,j)$ and $(i,j^\prime)$ with $j < j^\prime$ (according to ordering~\dref{order}), task $(i,j)$ appears before task $(i, j^\prime)$ in the list. On machine $i$, the algorithm scans the list starting from the first task. It schedules a task $(i,j)$ from the list if the machine has sufficient remaining resource to accommodate it. Upon completion of a task, the algorithm preempts the schedule, removes the completed task from the list and updates the remaining processing time of the tasks in the list, and starts scheduling the tasks in the updated list. Observe that this list scheduling is slightly different from the greedy scheme used in \AII.

A pseudocode for the algorithm can be found in Appendix~\ref{pseudocode2}. 
\subsection{\textbf{Performance Guarantee}}
The main result regarding the performance of {\AI} is summarized in the following theorem. The proof of the theorem, and any supporting lemmas, is presented in Appendix~\ref{Proof1}.

\begin{theorem}
	\label{theorem1}
	The scheduling algorithm {\AI} (Section~\ref{sec:firstpremeptive}) is a $4$-approximation algorithm for the problem of parallel-task jobs scheduling with packing and single-machine placement constraints.
\end{theorem}


\section{Complexity of Algorithms}
The complexity of our algorithms is mainly dominated by solving their corresponding LPs, which can be solved in polynomial time using efficient linear programming solvers. The rest of the operations have low complexity and can be parallelized on the machines. We have provided a detailed discussion of the complexity in Appendix \ref{complexity}.
\section{Evaluation Results}
\label{simulation}
In this section, we evaluate the performance of our algorithms using a real traffic trace, and compare to prior algorithms. 

\textbf{Data Set.} The data set is from a large Google cluster~\cite{clusterdata:Wilkes2011}. 
The original trace is 
over a month long period. To keep things simpler, we extract multi-task jobs of production scheduling class that were completed without any interruptions.In our experiments, we filter jobs and consider those with at most $200$ tasks, which constitute about $99 \%$ of all the jobs in the production class. Also, in order to have reasonable traffic density on each machine (since otherwise the problem is trivial), we consider a cluster with $200$ machines and randomly map machines of the original set to machines of this set. The final data set used for our simulations contains $7521$ jobs with an average of $10$ tasks per job. We also extracted memory requirement of each task and its corresponding processing time from the data set. In the data set, each job has a priority that represents its sensitivity to latency. There are $9$ different values of job priorities.

 In the original data set, each task is assigned to \emph{one} machine, and the scheduler needs to decide when to start its scheduling. This is similar to our model for our preemptive algorithm {\AI} in Section~\ref{specialcase}. To incorporate placement constraints, we modify the data set as follows. For each task, we randomly choose $3$ machines and assume that processing time of the task on these machines is equal to the processing time given in the data set. We allow the task to be scheduled on other machines; however, its processing time will be penalized by a factor $\alpha>1$. This is consistent with the data locality models in previous work (e.g.~\cite{wang2016maptask,{grandl2015multi}}).
 
We evaluate the performance of algorithms in both offline and online settings.  For the offline setting, we consider the first $1000$ jobs in the data set and assume all of these jobs are in the system at time $0$. For the online setting, all the $7521$ jobs arrive according to the arrival times information in the data set. Further, we consider $3$ different cases for weight assignments: 1) All jobs have equal weights, 2) Jobs are assigned random weights between $0$ and $1$, and $3$) Jobs' weights are determined based on the job priority and class information in the data set. 

\textbf{Algorithms.} We consider three prior algorithms, PSRS~\cite{schwiegelshohn2004preemptive}, Tetris~\cite{grandl2015multi}, and JSQ-MW~\cite{wang2016maptask} to compare with our algorithms {\AIII} and {\AI}. 
We briefly overview the algorithms below.

\textbf{1. \boldmath{\PSRS}~\cite{schwiegelshohn2004preemptive}}:
\textit{Preemptive Smith Ratio Scheduling} is a preemptive algorithm for the parallel task scheduling problem (see Section~\ref{relatedwork}) on a single machine. Modified Smith ratio of task $(i,j)$ is defined as $\frac{w_j}{a_{ij}p_{ij}}=\frac{w_j}{v_{ij}}$. Moreover, a constant $\calV=0.836$ is used in the algorithm. It also defines $T(a,t)$ to be the first time after $t$ at which at least $a$ amount of the machine's capacity is available, given the schedule at time $t$. 
On machine $i$, the algorithm first orders tasks based on the modified Smith ratio (largest ratio first). It then removes the first task $(i,j)$ in the list and as long as the task needs at most $50\%$ of the machine capacity $m_i$, it schedules the task in a non-preemptive fashion at the first time that available capacity of the machine is equal to or greater than the task's size, namely at $T(a_{ij},t)$ where $t$ is the current time and $a_{ij}$ is the size of task $(i,j)$. However, if task $(i,j)$ requires more than half of the machine's capacity, the algorithm determines the difference 
$T(a_{ij},t)-T(m_i/2,t)$. If this time difference is less than the ratio $p_{ij}/\calV$, it schedules task $(i,j)$
in the same way as those tasks with smaller size; that is, $(i,j)$ starts at $T(a_{ij},t)$ and runs to completion. Otherwise at time $T(m_i/2,t)+p_{ij}/\calV$, it preempts all the tasks that do not finish before that time, and starts task $(i,j)$. After task $(i,j)$ is completed, those preempted tasks are resumed. 
For the online setting, upon arrival of each task, the algorithm preempts the schedule, updates the list, and schedule the tasks in a similar fashion.

\textbf{2. \boldmath{$\mathsf{Tetris}$}~\cite{grandl2015multi}}: $\mathsf{Tetris}$ is a heuristic that schedules tasks on each machine according to an ordering based on their scores (Section~\ref{relatedwork}). Tetris was originally designed for the case that all jobs have identical weights; therefore, we generalize it by incorporating weights in tasks' scores.  For each task $(i,j)$ at time $t$, its score is defined as $s_{ij}=w_j(a_{ij}+\frac{\epsilon}{\sum_i a_{ij}p^t_{ij}})$, where $\epsilon=\frac{\sum_i\sum_j w_ja_{ij}}{\sum_j w_j (\sum_i a_{ij}p^t_{ij})^{-1}},$
and $p^t_{ij}$ is the task's remaining processing time at time $t$. 
Note that the first term in the score depends on the task' size (it favors a larger task if it fits in the machine's remaining capacity), while the second term prefers a task whose job's remaining volume (based on the sum of its remaining tasks) is smaller. On each machine, $\mathsf{Tetris}$ orders tasks based on their scores and greedily schedules tasks according to the list as far as the machine capacity allows. We consider two versions of $\mathsf{Tetris}$, preemptive ({\Tetp}) and non-preemptive ({\Tetnp}). In {\Tetp}, upon completion of a task (or arrival of a job, in the online setting), it preempts the schedule, update the list, calculate scores based on updated values, and schedule the tasks in a similar fashion. In {\Tetnp}, the algorithm does not preempt the tasks that are running; however, calculates scores for the remaining tasks based on updated values. 

To take the placement constraint into account, $\mathsf{Tetris}$ imposes a remote penalty to the computed score to penalize use of remote resources. This remote penalty is suggested to be $\approx 10\% $ in~\cite{grandl2015multi}. In simulations, we also simulated $\mathsf{Tetris}$ by penalizing scores by the factor $\alpha$, and found out that the performance is slightly better. Hence, we only report performance of $\mathsf{Tetris}$ with this remote penalty.

\textbf{3. \boldmath{$\mathsf{JSQ}\text{-}\mathsf{MW}$}~\cite{wang2016maptask}}: Join-the-Shortest-Queue routing with Max Weight scheduling (JSQ-MW) is a non-preemptive algorithm in presence of data locality (Section~\ref{relatedwork}). It assigns an arriving task to the shortest queue among those corresponding to the $3$ local servers with its input data and the remote queue. When a server is available, it either process a task from its local queue or from the remote  queue, where the decision is made based on a MaxWeight. 
%
	We further combine JSQ-MW with the greedy packing scheme so it can pack and schedule tasks non-preemptively in each server.

\textbf{4. \boldmath{\AIII} and \boldmath{\AI}}: These are our non-preemptive and preemptive algorithms as described in Section~\ref{nomigration} and Section~\ref{specialcase}. The complexity of our algorithms is mainly dominated by solving their corresponding LPs.  While (LP3) has reasonable size and can be solved quickly (see Section~\ref{complexity} for the details), (LP2) requires more memory for large instances. In this case, to expedite computation, besides the $3$ randomly chosen local machines that can schedule a task, we consider $10$ other machines ($5\%$ of the machines, instead of all the machines) that can process the task in an $\alpha$ times larger processing time. We choose these $10$ machines randomly as well. 
	Note that this may degrade the performance of our algorithm, nevertheless, as will see, they still significantly outperform the past algorithms.

A natural extension of our algorithms to online setting is as follows. 
We choose a parameter $\tau$ that is tunable. We divide time into time intervals of length $\tau$. For the preemptive case, at the beginning of each interval, we preempt the schedule, update the processing times, and run the offline algorithm on a set of jobs, consisting of jobs that are not scheduled yet completely and those that arrived in the previous interval. In the non-preemptive case, tasks on the boundary of intervals are processed non-preemptively, i.e., we let the running tasks (according to the previously computed schedule) finish, then apply the non-preemptive offline algorithm on the updated list of jobs as in the preemptive online case, and proceed with the new schedule. 
Note that a larger value of $\tau$ reduces the complexity of the online algorithm; but it also decreases the overall performance. We use an adaptive choice of $\tau$ to improve the performance of our online algorithm, starting from smaller value of $\tau$. In our simulations, we choose the length of the $i$-th interval, $\tau_i$, as $\tau_i=\tau_0/({1+\gamma \times \exp(-\beta i)}), \ i=1,2,\cdots$,
for some constants $\gamma$ and $\beta$. We choose $\tau_0=3 \times 10^2$ seconds, which is $5$ times greater than the average inter-arrival time of jobs, and $\gamma=50$ and $\beta=3$.

\begin{figure*}[t]
	\centering
	\begin{subfigure}{0.3\textwidth}
		\includegraphics[width=\textwidth]{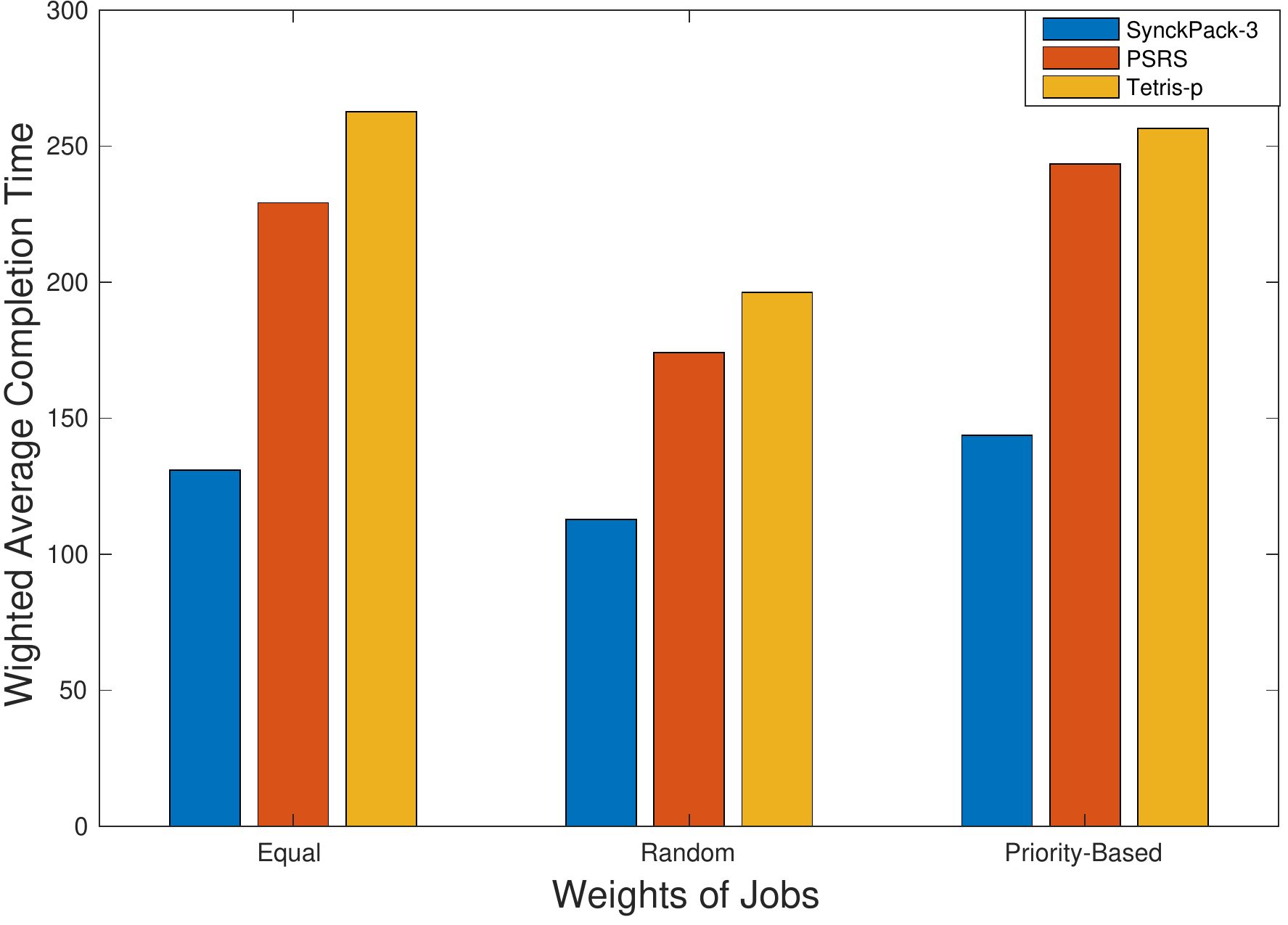}%
		\caption{Performance of \boldmath{\AI}, \boldmath{\Tetp}, and \boldmath{\PSRS} for different weights.}
		\vspace{10pt}
		\label{CT}%
	\end{subfigure}\hfill
	\begin{subfigure}{0.3\textwidth}
		\includegraphics[width=\textwidth]{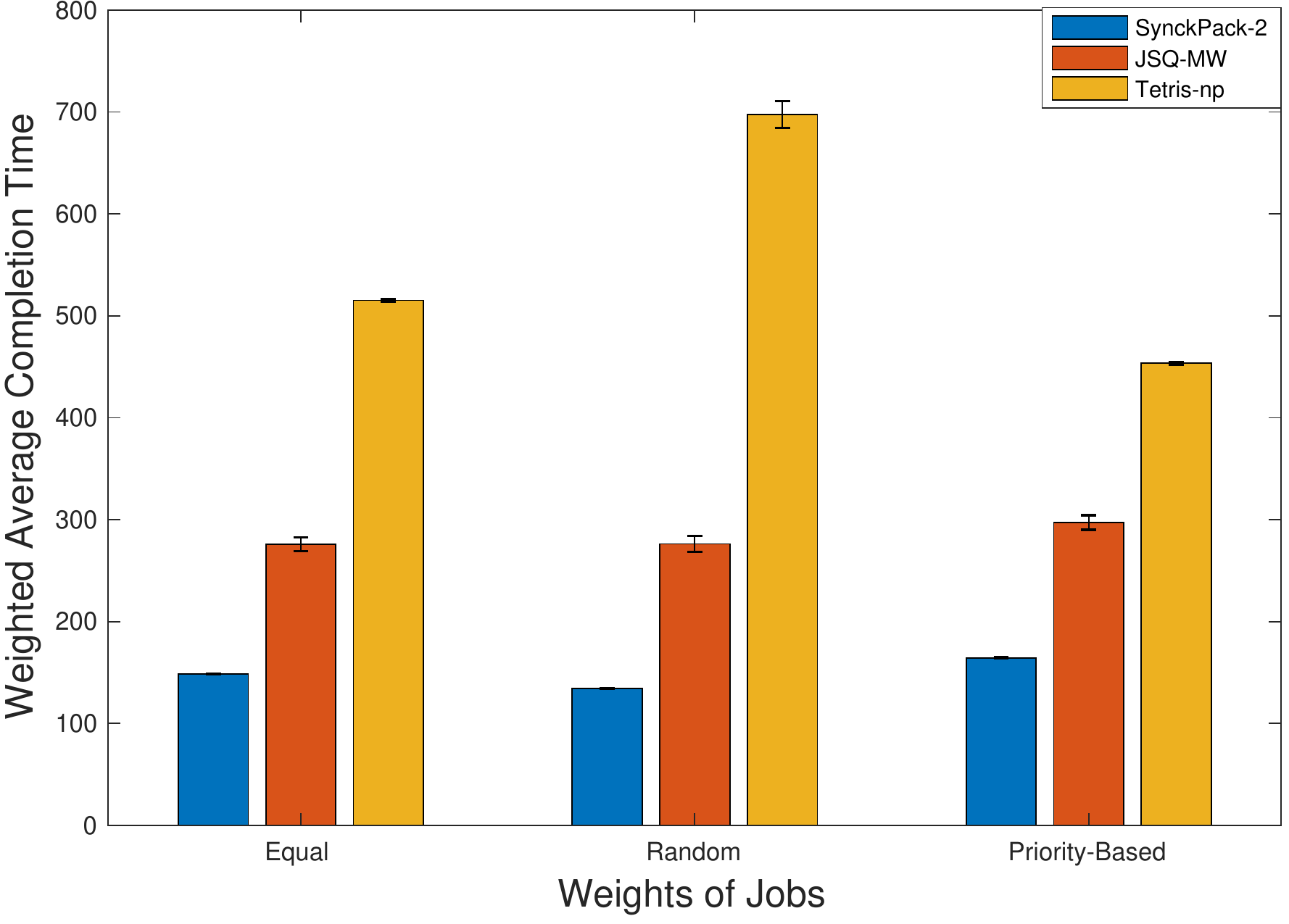}%
		\caption{Performance of \boldmath{\AIII}, \boldmath{\Tetnp}, and \boldmath{\JSQMW} for different weights and remote penalty $\alpha=2$.}
		\label{CT4alpha2}%
	\end{subfigure}\hfill
	\begin{subfigure}{0.3\textwidth}
		\includegraphics[width=\textwidth]{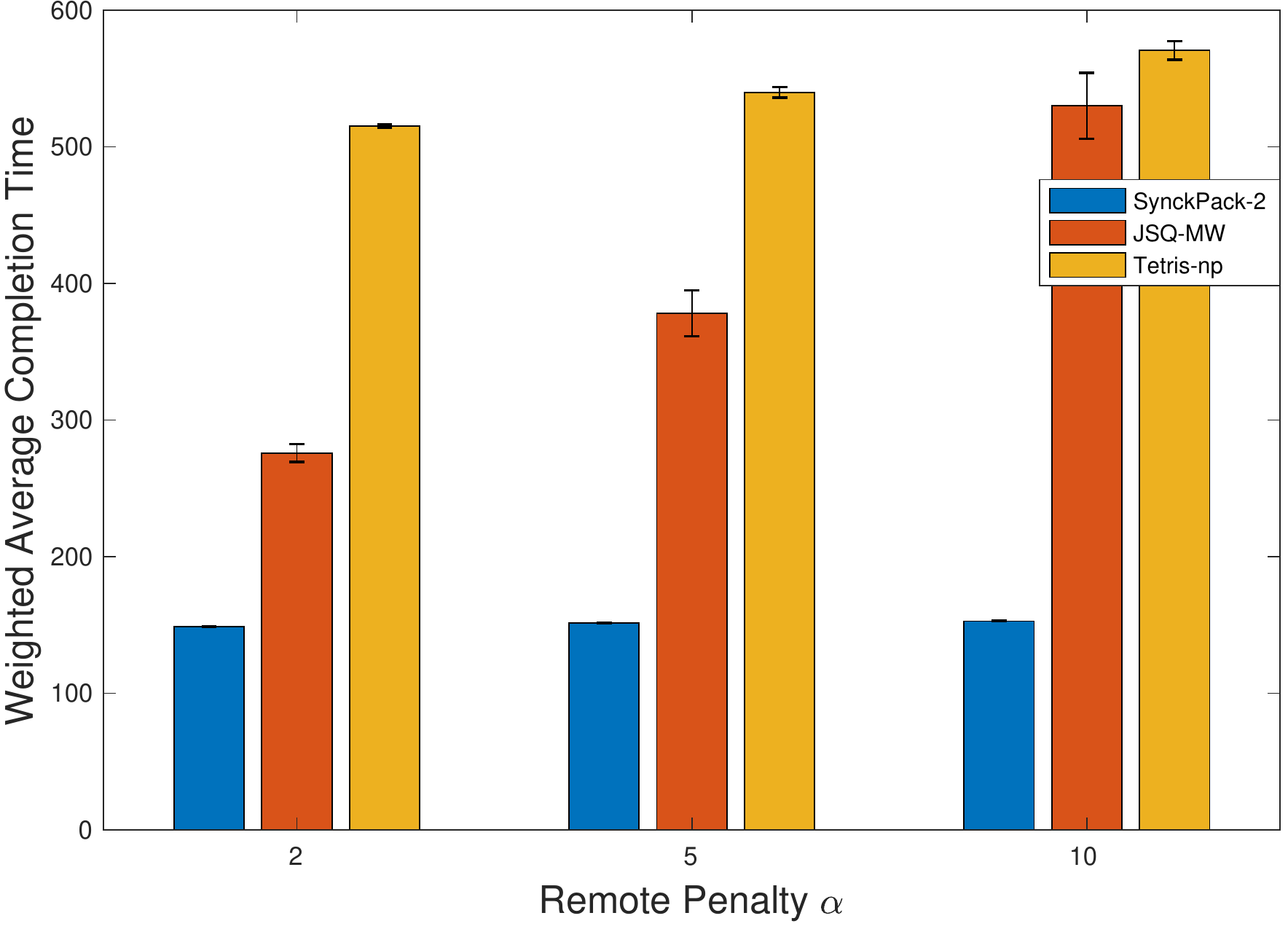}%
		\caption{Performance of \boldmath{\AIII}, \boldmath{\Tetnp}, and \boldmath{\JSQMW} for different remote penalties and equal weights.}
		\label{CTequal}%
	\end{subfigure}
	\vspace{-0.1 in}
	\caption{Performance of algorithms in the offline setting.}
\end{figure*}
\begin{figure*}[t]
		\vspace{-0.05 in}
	\centering
	\begin{subfigure}{0.3\textwidth}
		\includegraphics[width=\textwidth]{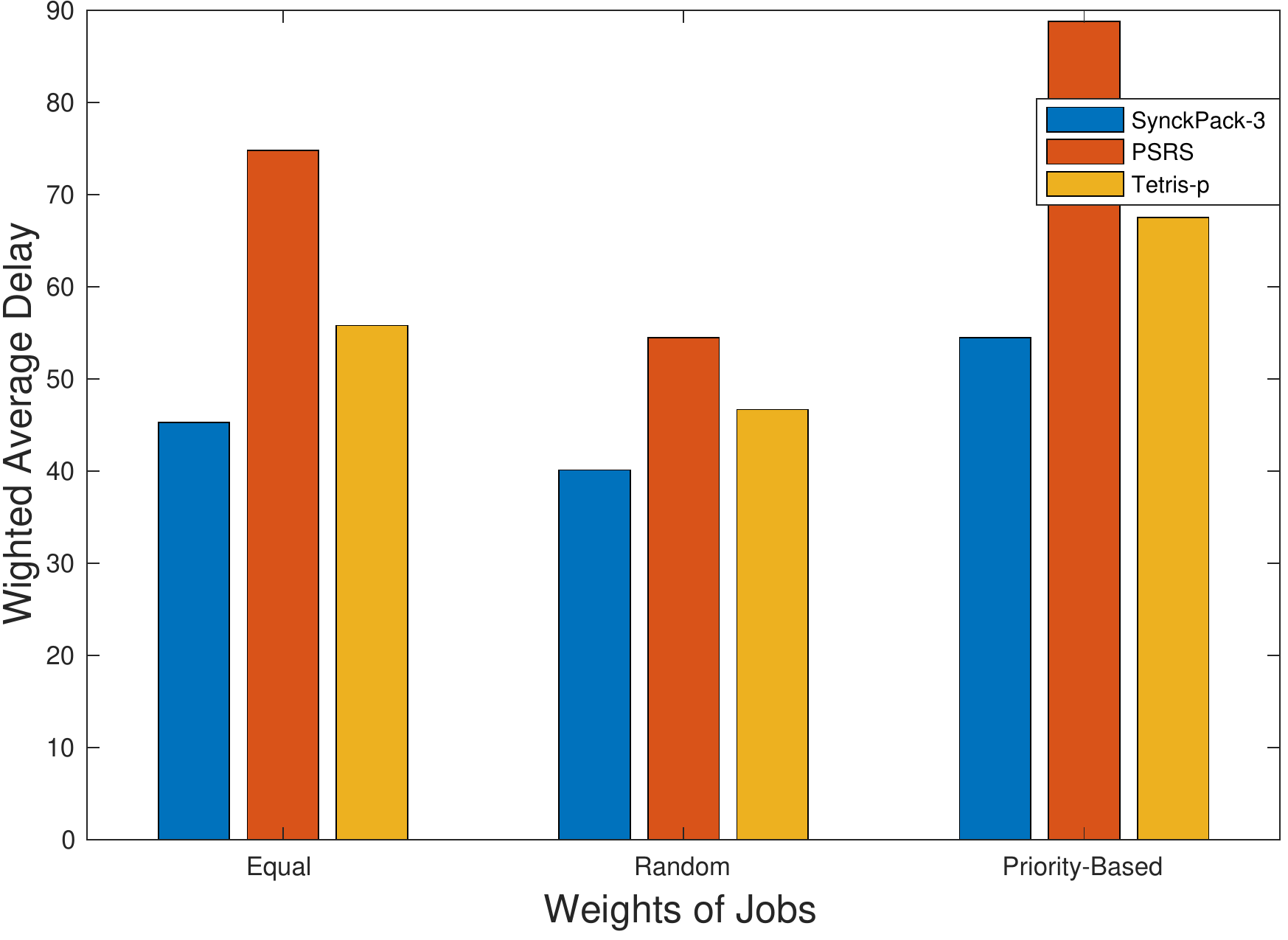}%
		\caption{Performance of \boldmath{\AI}, \boldmath{\Tetp}, and \boldmath{\PSRS} for different weights.}
		\label{delay}%
	\end{subfigure}\hfill
	\begin{subfigure}{0.3\textwidth}
		\includegraphics[width=\textwidth]{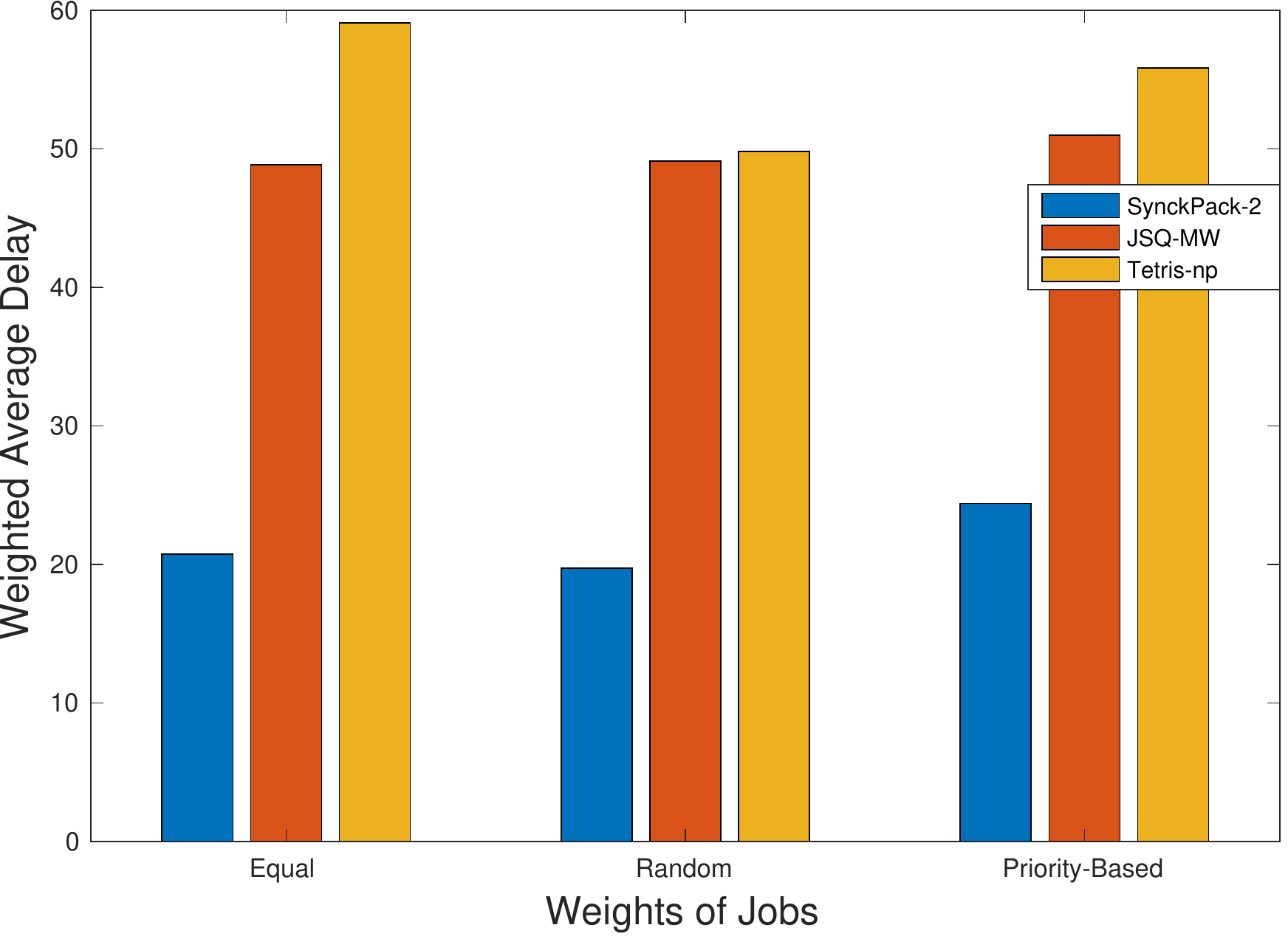}%
		\caption{Performance of \boldmath{\AIII}, \boldmath{\Tetnp}, and \boldmath{\JSQMW} for different weights.}
		\label{delaynp}%
	\end{subfigure}\hfill
	\begin{subfigure}{0.3\textwidth}
		\includegraphics[width=\textwidth]{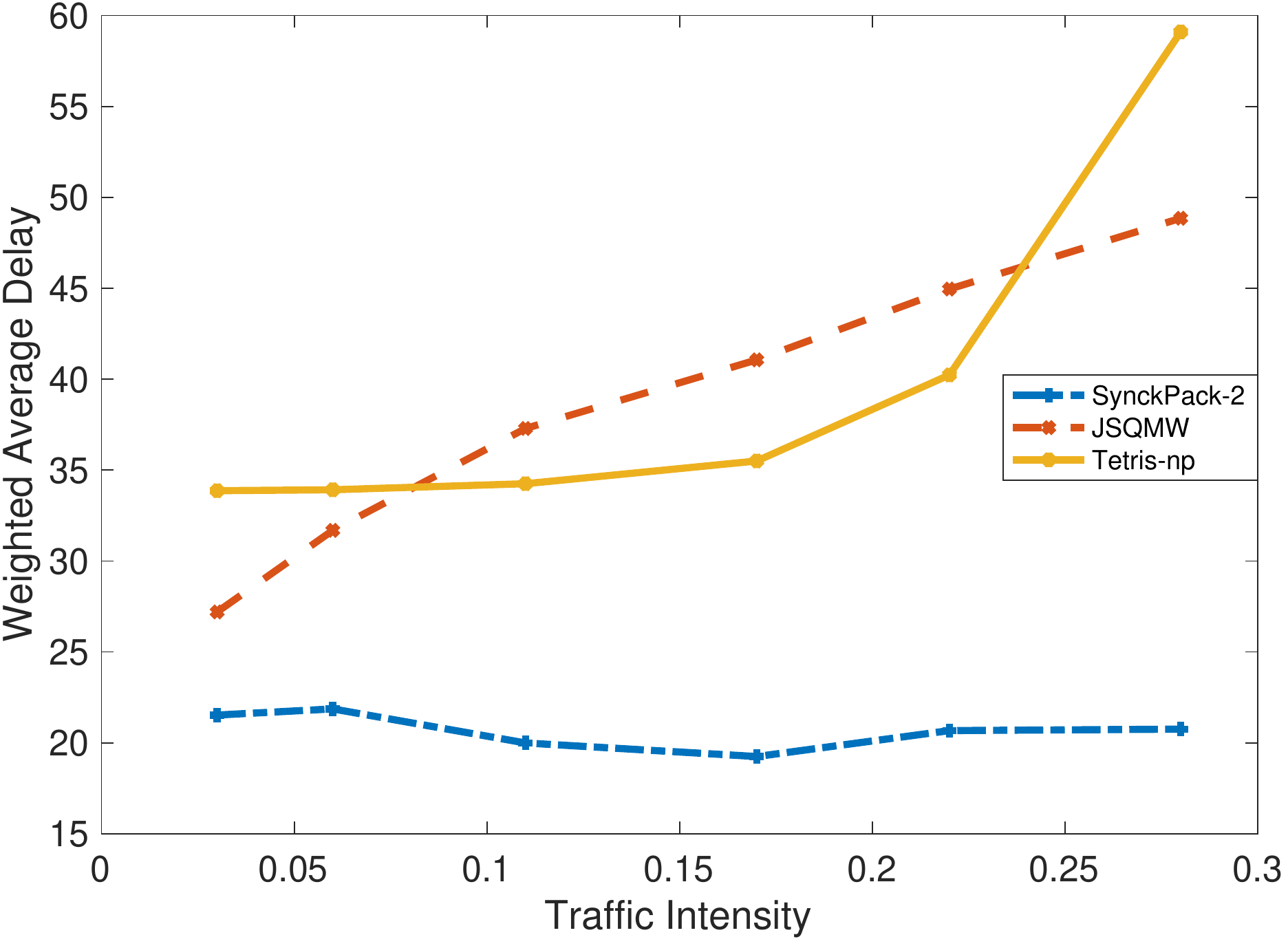}%
		\caption{Performance of \boldmath{\AIII}, \boldmath{\Tetnp}, and \boldmath{\JSQMW} for different traffic intensity.}
		\label{traffint}%
	\end{subfigure}
	\vspace{-0.1 in}
	\caption{Performance of algorithms in the online setting.}

\end{figure*}

\subsection{Results in Offline Setting}
We use \AI, \Tetp, and {\PSRS} to schedule tasks of the original data set preemptively, and use \AIII, \Tetnp, and {\JSQMW} to schedule tasks of the modified data set (with placement constraints) non-preemptively. We then compare the weighted average completion time of jobs, 
$\sum_j w_jC_j/ \sum_j w_j$, 
under these algorithms for the three weight cases, i.e. equal, random, and priority-based weights. 
Note that weighted average completion time is equivalent to the total weighted completion time (up to the normalization $\sum_j w_j$).

We first report the ratio between the total weighted completion time obtained from {\AIII} (for $\alpha=2$) and {\AI} and their corresponding optimal value of their relaxed LPs~\dref{RLP3} and~\dref{RLP1} (which is a lower bound on the optimal total weighted competition time) to verify Theorem~\ref{theorem2} and~\ref{theorem1}.
Table~\ref{ratio} shows this performance ratio for the $3$ cases of job weights. All ratios are within our theoretical approximation ratio of $24$ and $4$. In fact, the approximation ratios are much smaller.
\begin{table}[h]
	\caption{Performance ratio between {\boldmath{\AI}} and (LP3), and between {\boldmath{\AIII}} and (LP2)}
	\centering
	\begin{tabular}{c c c c} 
		\hline\hline 
		Jobs' Weights & Equal & Random & Priority-Based \\ [0.4ex] 
		\hline 
		Ratio for \AIII & $2.87$ & $2.90$ & $2.98$\\
		\hline
		Ratio for \AI & $1.34$ & $1.35$ & $1.31$\\
		\hline 
	\end{tabular}
	\label{ratio}
\end{table}

Figure~\ref{CT} shows the performance of {\AI}, {\Tetp}, and {\PSRS} in the offline setting. As we see, {\AI} outperforms the other two algorithms in all the cases and performance gain varies from $33\%$ to $132\%$. 
Further, Figure~\ref{CT4alpha2} depicts performance of {\AIII}, {\Tetnp}, and {\JSQMW} for different weights, when $\alpha=2$. The performance gain of {\AIII} varies from $81\%$ to $420\%$. Figure~\ref{CTequal} shows the effect of remote penalty $\alpha$ in the performance of {\AIII}, {\Tetnp}, and {\JSQMW}.  As we see, {\AIII} outperforms the other algorithms by $85\%$ to $273\%$	
	
\subsection{Results in Online Setting}
In the online setting, jobs arrive dynamically over time, according to the arrival time information in the data set, and we are interested in the weighted average \emph{delay} of jobs. The delay of a job is measured from the time that it arrives to the system until its completion.

Figure~\ref{delay} shows the performance results, in terms of the weighted average \emph{delay} of jobs, under \AI, \Tetp, and {\PSRS}. Performances of {\Tetp} is worse than our algorithm by $11\%$ to $27\%$, while {\PSRS} presents the poorest performance and has $36\%$ to $65\%$ larger weighted average delay compared to \AI. 	
Moreover, performance of {\AIII}, {\Tetnp}, and {\JSQMW} for different weights is depicted in Figure~\ref{delaynp}. As we see, {\AIII} outperforms the other two algorithms in all the cases and performance gain varies from $109\%$ to $189\%$. 

Further, by multiplying arrival times by constant values we can change the traffic intensity and study its effect on algorithms' performance. Figure~\ref{traffint} shows the results for equal job weights. As we can see, {\AIII} outperforms the other algorithms and the performance gain increases as traffic intensity grows.


%

\section{Conclusions}
\label{conclusion}
We studied the problem of scheduling jobs, each job with multiple resource constrained tasks, in a cluster of machines. Our motivation for this model came from modern parallel computing applications where tasks can be packed in machines, the input data required to process each task is stored in a subset of machines, and the tasks' collective completion time determines the job's completion.  
We proposed the first constant-approximation algorithms for minimizing the total weighted completion time of such jobs, namely, a $(6+\epsilon)$-approximation preemptive algorithm when task migration is allowed, and a $24$-approximation non-preemptive (without migrations) algorithm.  Further, for the special case when there is a specific machine for each task (single-machine data locality), we improved the approximation ratio for the preemptive case to $4$. The model and analysis in our setting of tasks with packing,  synchronization, and placement constraints is new. Note that the approximation results are upper bounds on the algorithms' performance. Our simulation results for our algorithms showed that the approximation ratios are in fact very close to $1$ in practice.


As we showed, applying our simple greedy packing, to schedule tasks mapped to each interval in {\AIII}, provides a tight bound on the total volume of tasks and its relation to the associated linear program. Therefore, we cannot improve the final result by replacing this step with more intelligent bin packing algorithms like BestFit~\cite{coffman1980performance}. Although, in practice, applying such bin packing schemes can give a better performance. Improving the performance bound of $24$ requires a more careful and possibly different analysis. We leave further improvement of the result as a future work.
Extension of our model to capture multi-dimensional task resource requirements ,and analysis of the online algorithms for our problem, are also interesting and challenging topics for future work. 

\newpage
\bibliographystyle{ACM-Reference-Format}
\bibliography{Bibl}

\section*{Appendix}
\appendix

\section{Complexity of Algorithms} 
\label{complexity}

The linear program (LP1) in \dref{RLP2} has at most $KNML+NL +N$  variables ($K$ is the maximum number of tasks a job has.),
 which is polynomially bounded in the problem's input size. The number of constraints is also polynomially bounded. Hence, it can be solved in polynomial time using efficient linear programming solvers.
The complexity of {\AII} is mainly determined by solving (LP1). The complexity of Slow-Motion step is very low and can be parallelized in different machines, namely, $O(KNL)$ on each machine, and $O(KNLM)$ in total. The complexity of the greedy list scheduling -- upon arrival or departure of a task fraction-- is at most the length of the list (equal to the number of incomplete task fractions which is initially equal to $O(KNLM)$) times the number of machines $M$.

Mapping procedure is the extra step for {\AIII}. The complexity of this step is also polynomially bounded in input size and is $O(K^2N^2ML)$. $O(KN+ML)$ is used for constructing the graph as there are $O(KN)$ nodes on one side (number of all the tasks), $O(KN+ML)$ on the other side (number of all machine-interval copies), and it takes $O(KNML)$ to create edges (each task has at most $2$ edges to copies of each machine-interval.). Further, finding an integral matching from the fractional matching takes $O(K^2N^2ML)$. The greedy algorithm in {\AIII} can be parallelized on the machines and takes $O(KN)$ in total.

Similarly, the complexity of {\AI} is mainly dominated by solving (LP3) to find an appropriate ordering of jobs. The relaxed linear program (LP3) has $O(N^2)$ variables and $O(N^2+MN)$ constraints and can be solved in polynomial time using efficient linear programming solvers. Note that the job ordering is the same on all the machines and they simply list-schedule their tasks respecting this ordering, independently of other machines. The complexity of the list scheduling is less than the one used in {\AIII} and is at most the length of the list, which is equal to the number of incomplete tasks. 
 
Further, we would like to emphasize that in all the algorithms the corresponding linear program (LP) is solved \textit{only once} at the beginning of the algorithm.  

For the simulations, we used Gurobi software~\cite{gurobi} to solve (LP2) and (LP3) in the simulations. On a desktop PC, with $8$ Intel CPU core $i7-4790$ processors @ $3.60$ GHz and $32.00$ GB RAM, the average time it took to solve (LP1) was $145$ seconds under offline setting. For purpose of comparison, the maximum job completion and the weighted average completion time time under our algorithm are $4.3 \times 10^4$ seconds and $8.6 \times 10^3$ seconds, respectively, for the case of priority-based weights. 
For solving (LP3), the average time it took was $435$ seconds under offline setting, while the maximum job completion time and the weighted average completion time under our algorithm are $4.8 \times 10^4$ seconds and $10^4$ seconds, respectively for the case of priority-based weights for $\alpha=2$. 
We note that solving the LPs can be done much faster using the powerful computing resources in today's datacenters.

\section{Proofs Related to \lowercase{\boldmath}{\AII}}
\label{Proof2}
\subsection{Proof of Lemma~\ref{lp3bound}}
\label{Proof2-1}
	Consider an optimal solution to the task scheduling problem with packing and synchronization constraints. Define $\hat{C}^\star_{kj}$ (similarly, $\hat {C}_{j}^\star$) to be the left point of the interval in which task $(k,j)$ (similarly, job $j$) completes in the optimal schedule. Clearly, $\hat{C}_{j}^\star \leq C^\star_j$. We set ${z_{kj}^{{il}^\star}}$ equal to the fraction of task $(k,j)$ that is scheduled in interval $l$ on machine $i$. Also, we set $x^\star_{j,l}$ to be one for the last interval that some task of job $j$ is running in the optimal schedule and to be zero for other intervals. Obviously, $\hat {C}_{j}^\star = \sum_{l=0}^L d_{l-1} x^\star_{jl}$. It is easy to see that the set of values $\hat {C}_{j}^\star$, ${z_{kj}^{il}}^\star$, and $x^\star_{j,l}$ satisfies all the constraints of (LP3). Therefore, $\sum_{j=1}^N w_j \tilde{C}_j \leq \sum_{j=1}^N w_j \hat{C}^\star_j \leq \sum_{j=1}^N w_j C^\star_j$.
	
\subsection{Proof of Lemma~\ref{intcjalpha}}
\label{Proof2-5}
	Recall that $\tau_l$ is the time that all the task fractions $(k,j,i,l^\prime)$, for $l^\prime \leq l$, complete in schedule $\calS$. Let $\alpha_l$ be the fraction of job $j$ that is completed by $\tau_l$.
	
	Note that as we schedule all the task fractions $(k,j,i,l^\prime)$, for $l^\prime \leq l$ and possibly some other task fractions, we have,
	\begin{equation}
	\label{alphadef}
	\alpha_l \geq \sum_{l^\prime=0}^l \tilde{x}_{jl^\prime}.
	\end{equation}
	We define $y_{jl}=\alpha_{l}-\alpha_{l-1}$. Note that $\sum_{l=0}^L y_{jl}=1$. Moreover, $C_j(\alpha) \leq 3 d_{l}$ for $\alpha \in (\alpha_{l-1}, \alpha_{l}]$. The factor $3$ comes from Lemma~\ref{tripledinterval}. Therefore:
	\begin{equation}
	\label{equality}
	\begin{aligned}
	\int_0^1 C_j(\alpha) d\alpha &= \sum_{l=0}^L \int_{\alpha_{l-1}}^{\alpha_{l}} C_j(\alpha) d\alpha \leq \sum_{l=0}^L (\alpha_{l}-\alpha_{l-1}) \times 3 d_{l}\\
	&\stackrel{(a)}=3 (1+\epsilon)\sum_{l=0}^L y_{jl} d_{l-1}\stackrel{(b)}\leq 3 (1+\epsilon)\sum_{l=0}^L \tilde{x}_{jl} d_{l-1}\\
	&\stackrel{(c)}=3 (1+\epsilon) \tilde{C_j},
	\end{aligned}
	\end{equation}
	where (a) follows from definitions. Inequality (b) follows from \dref{alphadef} when $y_{jl}$ and ${x}_{jl}$ is seen as probabilities. Equality (c) comes from \dref{cdef} in (LP1).

\subsection{Proof of Lemma~\ref{smbound}}
\label{Proof2-6}
	It is easy to observe that for every job $j$, $\bar{C}_j^\lambda \leq C_j(\lambda)/\lambda$. The reason is that $C_j(\lambda)$ is the time that $\lambda$ fraction of job $j$ is completed in $\calS$; therefore, in the stretched schedule $\bar{\calS}$ by factor $1/\lambda$, job $j$ is completed by time $C_j(\lambda)/\lambda$. Hence, we have
	\begin{equation*}
	\begin{aligned}
	\expect{\bar{C}_j^\lambda} &\leq \expect{C_j(\lambda)/\lambda} \stackrel{(a)} =\int_0^1 \frac{C_j(\lambda)}{\lambda} \times 2\lambda \times d\lambda\\& \stackrel{(b)} \leq 6(1+\epsilon)\tilde{C_j},
	\end{aligned}
	\end{equation*}
	where Equality (a) is by definition of expectation with respect to $\lambda$, with pdf $f(\lambda)=2\lambda$, and Equality (b) is due to Lemma~\ref{intcjalpha}.
\section{De-randomization} \label{derand} In this section, we discuss how to de-randomize the random choice of $\lambda \in (0, 1]$ in \AII, which was used to construct schedule $\bar{\calS}$ from schedule $\calS$.

Recall that from Definition~\ref{calpha}, $C_j(\lambda)$, $0 < \lambda \leq 1$, is the starting point of the earliest interval in which $\lambda$-fraction of job $j$ has been completed in schedule $\calS$, which means at least $\lambda$-fraction of each of its tasks has been completed. We first aim to show that we can find \begin{equation}
\label{lambdastar}
\lambda^\star=\argmin_{\lambda \in (0,1]} \sum_{j\in \calJ} w_jC_j(\lambda)/\lambda
\end{equation}
in polynomial time. Note that using the greedy packing algorithm, we schedule task fractions preemptively to form schedule $\calS$. 
It is easy to see that $C_j(\lambda)$ is a step function with at most $O(L)$ breakpoints, since $C_j(\lambda)=d_l$ for some $l$ and can get at most $L$ different values. Consequently, $F(\lambda)=\sum_{j\in \calJ} w_jC_j(\lambda)$ is a step function with at most $O(NL)$ breakpoints. Let $B$ denote the set of breakpoints of $F(\lambda)$. Thus, $F(\lambda)/\lambda=\sum_{j\in \calJ} w_jC_j(\lambda)/\lambda$ is a non-increasing function in intervals $(b,b^\prime]$, for $b, b^\prime$ being consecutive points in set $B$. This implies that,
\begin{equation*}
\min_{\lambda \in (0,1]}F(\lambda)/\lambda=\min_{\lambda \in (0,1]} \sum_{j\in \calJ} w_jC_j(\lambda)/\lambda=\min_{\lambda \in B} \sum_{j\in \calJ} w_jC_j(\lambda)/\lambda.
\end{equation*}
We then can conclude that we can find $\lambda^\star$ in polynomial time by checking values of function $F(\lambda)/\lambda$ in at most $O(NL)$ points of set $B$ and pick the one which incurs the minimum value. Given that, we have
\begin{equation}
\label{derandlambda}
\begin{aligned}
\sum_{j \in \calJ} w_j \bar{C}_j^{\lambda^\star} &\leq \sum_{j \in \calJ} (1+\epsilon) w_j C_j(\lambda^\star)/\lambda^\star\\
&\stackrel{(a)}\leq (1+\epsilon) \expect{\sum_{j \in \calJ} w_jC_j(\lambda)/\lambda}\\
&=(1+\epsilon)\sum_{j \in \calJ} w_j \int_{\lambda=0}^{1}\frac{C_j(\lambda)}{\lambda} 2\lambda d\lambda\\
&\stackrel{(b)}=6(1+\epsilon)\sum_{j \in \calJ} w_j\tilde{C_j},
\end{aligned}
\end{equation}
where (a) follows from~\dref{lambdastar}. Equality (b) is due to Lemma~\ref{intcjalpha}. By choosing $\lambda=\lambda^\star$ in \AII, we have a deterministic algorithm with performance guarantee of $(6+\epsilon)\times \text{OPT}$.
, as stated by the following proposition.
\section{Proofs Related to \lowercase{\boldmath} \AIII}
\label{Proof3}
\subsection{Proof of Lemma~\ref{lp4bound}}
\label{Proof3-1}
Consider an optimal solution to the non-preemptive task scheduling problem with packing and synchronization constraints. For each task, we set ${z_{kj}^{il}}^\star=1$ for the machine $i$ and interval $l$ if that task $(k,j)$ is processed on $i$ and finishes before $d_l$, and $0$ otherwise. The rest of argument is similar to the proof of Lemma~\ref{lp3bound}.

\subsection{Proof of corollary~\ref{coro}}
\label{Proof3-4}
	Note that (LP2) includes all the Constraints~\dref{xdef}--\dref{jobdemand} of (LP1). Let $\alpha_l$ be the fraction of job $j$ that is completed by interval $l$. Therefore,
	\begin{equation}
	\label{equal1}
	\alpha_l=\sum_{l^\prime=0}^l \tilde{x}_{jl^\prime}.
	\end{equation}
	Similar to Equations~\dref{equality}, we can write
	\begin{equation*}
	\begin{aligned}
	\int_0^1 \hat{C}_j(\alpha) d\alpha= \sum_{l=0}^L (\alpha_{l}-\alpha_{l-1}) \times d_{l-1}=\sum_{l=0}^L \tilde{x}_{jl} d_{l-1}=\tilde{C_j},
	\end{aligned}
	\end{equation*}	

\subsection{Proof of Lemma~\ref{intmatching}}
\label{Proof3-2}
	We use the following fundamental theorem (Theorem 2.1.3 in \cite{scheinerman2011fractional}): If there exists a fractional matching of some value $\nu$ in a bipartite graph $G$, then there exists an integral matching of the same value $\nu$
	in $G$ on the non-zero edges and can be found in polynomial time.
	
	In our constructed bipartite graph $\calG$, edge weights $w_{kj}^{ilc}$ can be seen as a fractional matching. This is because for any node $u \in U$, the sum of weights of edges that are incident to $u$ is  $1$, while for any node $v \in V$ the sum of weights of edges that are incident to $v$ is at most $1$. Recall that $|\cup_{j \in \calJ} \calK_j|=\sum_{j \in \calJ} \sum_{k \in \calK_j} \sum_{l=0}^L \bar{z}_{kj}^{il}$ is the number of total tasks. Setting $G=\calG$ and $\nu=|\cup_{j \in \calJ} \calK_j|$, an integral matching of nodes in $U$ to nodes in $V$ on non-zero edges can be found in polynomial time by the stated theorem.

\subsection{Proof of Lemma~\ref{machineintervalvol}}
\label{Proof3-3}
We now present the proof of Lemma~\ref{machineintervalvol} which bounds $V_{il}$ (the total volume of tasks matched to all copies of interval $l$ for machine $i$) by the product of the capacity of machine $i$ and the length of interval $l$.
Observe that due to definition of $v_{kj}^i$ and Constraint~\dref{intcap} we have,
\begin{equation}
\label{loadsbar}
\sum_{j \in \calJ} \sum_{k \in \calK_j} {v_{kj}^i}\bar{z}_{kj}^{il} \leq \bar{d}_l m_i,
\end{equation}
The proof idea is similar to~\cite{lenstra1990approximation} that uses a simpler version of the mapping procedure in makespan minimization problem for scheduling tasks with unit resource requirements on unrelated machines with unit capacities, where each task can be scheduled in any machine. Let $V_{il}^c$ denote the volume of the task that is matched to copy $c$ of interval $l$ on machine $i$. Thus, $V_{il}$ is equal to the sum of $V_{il}^c$ for all copies. Recall that we have $\lceil \bar{z}^{il} \rceil=\lceil \sum_{j \in J} \sum_{k \in \calK_j} \bar{z}_{kj}^{il} \rceil$ many copies of interval $l$. Let $V_{il}^{\max}$ denote the largest volume of the task that is mapped to interval $l$. For this task, we know that $\bar{z}_{kj}^{il} > 0$ because the integral matching was found on nonzero edges (line 23 in Algorithm~\ref{subroutine}); hence, $p_{kj}^i \leq d_l = \lambda \bar{d}_l \leq \bar{d}_l$ by Constraint~\dref{extraconst} and $\lambda \in (0,1]$. In addition, let $v_{il}^{\min_c}$ denote the volume of the smallest task that has an edge with non-zero weight to copy $c$ of interval $l$ in graph $\calG$ (or equivalently, has a non-zero edge in the fractional matching.). Observe that, the volume of the task that is matched to copy $c+1$ is at most $v_{il}^{\min_c}$. This is because of the way we construct graph $\calG$ by sorting tasks according to their volumes for each machine (see the ordering in~\dref{order3}) and the way we assign weights to edges. Thus,
	\begin{equation*}
	\begin{aligned}
	V_{il} &= \sum_{c=1}^{\lceil \bar{z}^{il} \rceil} V_{il}^c \leq V_{il}^{\max}+\sum_{c=2}^{\lceil \bar{z}^{il} \rceil} v_{il}^{\min_{c-1}} \\& \stackrel{(a)} \leq \bar{d}_l m_i + \sum_{j \in \calJ} \sum_{k \in \calK_j} \sum_{c=1}^{\lceil \bar{z}^{il} \rceil-1} v_{kj}^i w_{kj}^{ilc}.
	\end{aligned}
	\end{equation*}
	Inequality (a) comes from the fact that $\sum_{j \in J}\sum_{k \in \calK_j}  w_{kj}^{ilc} \leq 1$ and convex combination of some numbers is greater than the minimum number among them (note that the only copy for which we might have $\sum_{j \in J} \sum_{k \in \calK_j} w_{kj}^{ilc} < 1$ is the last copy which is not considered in the left hand side of Inequality (a)). Therefore, as the direct result of the way we constructed graph $\calG$, we have
	\begin{equation*}
	\begin{aligned}			
	V_{il} \leq \bar{d}_l m_i + \sum_{j \in \calJ} \sum_{k \in \calK_j} v_{kj}^i \bar{z}_{kj}^{il} 
	\end{aligned}
	\end{equation*}

\subsection{Proof of Lemma~\ref{doubledinterval}}
\label{Proof2-2}
Lemma~\ref{doubledinterval} ensures that we can accommodate all the task fractions mapped to machine-interval $(i,l)$ within an interval with length twice $d_l+\sum_{j \in \calJ} \sum_{k \in \calK_j} v_{kj}^i \bar{z}_{kj}^{il}/m_i$.

	Similar to Definition~\ref{machineheight}, we define $h(t)$ to be the height of the machine at time $t$. Assume that completion time of the last task, $\tau$, is larger than $2V=2\max(1,v)$, then
	\begin{equation*}
	\sum_{j \in J} a_j p_j=\int_0^\tau h(t) dt > \int_0^{2V} h(t)dt \geq \int_0^{V} (h(t)+h(t+1))dt>1+v,
	\end{equation*}
	where we have used the fact that $h(t)+h(t+1)>1$, because otherwise the greedy scheduling can move tasks from time $t+1$ to time $t$ as the greedy scheduling is non-preemptive and $p_j \leq 1$ for all tasks. Hence we arrived at a contradiction and the statement of Lemma~\ref{doubledinterval} indeed holds.

\subsection{Proof of Lemma~\ref{tight}}
\label{Proof2-3}
Let $\max(1,v)=1$. We show correctness of Lemma~\ref{tight} by constructing an instance for which an interval of size at least $2-\zeta$ is needed to be able to schedule all the tasks for any $\zeta>0$. Given a  $\zeta>0$, consider $n > log_2 (1/\zeta)+1$ tasks with processing times $1,1/2,1/4,\dots,1/2^{(n-1)}$ and size $1/2+\eta$, for some $\eta>0$ which is specified shortly. Note that we cannot place more than one of such tasks at a time on the machine, and therefore we need an interval of length $1+1/2+1/4+\dots+1/2^{(n-1)}=2-1/2^{(n-1)}>2-\zeta$ to schedule all the tasks. The total volume of tasks is equal to $(1/2+\eta)(2-1/2^{(n-1)})$ which is less than $1$, by choosing $\eta \leq 1/(2^{(n+1)}-2)$. Therefore, for any $\zeta>0$, we can construct an example for which an interval of length at least $2-\zeta$ is needed to schedule all the tasks.

\section{Proofs Related to \lowercase{\boldmath}\AI}
\label{Proof1}
This section is devoted to the proof of the Theorem~\ref{theorem1}. We first characterize the solution of the linear program (LP3). 
\begin{lemma}
	\label{ineq1}
	Let $\tilde{C}_j$ be the optimal solution to (LP3) for completion time of job $j$, as in the ordering \dref{order}. For each machine $i$ and each job $j$, $m_i \tilde{C}_j \geq \frac{1}{2} \sum_{k=1}^j v_{ik}$.
\end{lemma}
\begin{proof}
	\label{Proof1-1}
	Using Constraint~\dref{volume}, for any machine $i \in \calM$, we have
	\begin{equation*}
	\begin{aligned}
	v_{ij} m_i \tilde{C_j} \geq  v_{ij}^2+ \sum_{j^\prime \in \calJ, j^\prime \neq j } v_{ij} v_{ij^\prime} \delta_{j^\prime j}.
	\end{aligned}
	\end{equation*}
	Hence, by defining $\delta_{kk}=0$, it follows that
	\begin{equation}
	\begin{aligned}
	\label{help}
	\sum_{k=1}^j v_{ik} m_i \tilde{C_k} \geq \frac{1}{2} \bigg ( 2 \sum_{k=1}^j v_{ik}^2
	+ \sum_{k=1}^j \sum_{k^\prime=1}^j \big (v_{ik} v_{ik^\prime} \delta_{k^\prime k}+v_{ik} v_{ik^\prime} \delta_{k k^\prime} \big ) \bigg )
	\end{aligned}
	\end{equation}
	We simplify the right-hand side of \dref{help}, using Constraint~\dref{prec1}, combined with the following equality
	\begin{equation*}
	\sum_{k=1}^j v_{ik}^2 + \sum_{k=1}^j \sum_{\substack{k^\prime=1 \\ k^\prime \neq k}}^j v_{ik} v_{ik^\prime} =(\sum_{k=1}^j v_{ik})^2,
	\end{equation*}
	and get
	\begin{equation}
	\label{capcons}
	\begin{aligned}
	\sum_{k=1}^j v_{ik} m_i \tilde{C_k} \geq \frac{1}{2} \sum_{k=1}^j (v_{ik})^2 + \frac{1}{2} (\sum_{k=1}^j v_{ik})^2 \geq \frac{1}{2} (\sum_{k=1}^j v_{ik})^2.
	\end{aligned}
	\end{equation}
	Given that $\tilde{C_j} \geq \tilde{C_k}$ for $1 \leq k \leq j$, we get the final result.
\end{proof}
Let $C^\star_j$ be the completion time of job $j$ in an optimal schedule, and $\text{OPT}=\sum_{j=1}^N w_j C^\star_j$ be the optimal value of our job scheduling problem.  
The following lemma states that the optimal value of (LP3), i.e., $\sum_{j=1}^N w_j \tilde{C}_j$, is a lower bound on the optimal value $\text{OPT}$.
\begin{lemma}
	\label{bound}
	$\sum_{j=1}^N w_j \tilde{C}_j \leq \sum_{j=1}^N w_j C^\star_j = \text{OPT}$. 
\end{lemma}
\begin{proof}
	\label{Proof1-2}
	Consider an optimal preemptive solution to the task scheduling problem with packing and synchronization constraints. We set the ordering variables such that $\delta_{jj^\prime}=1$ if job $j$ precedes job $j^\prime$ in this solution, and $\delta_{jj^\prime}=0$, otherwise. We note that this set of ordering variables and job completion times satisfies Constraint~\dref{volume} since this solution will respect resource constraints on the machines. It also satisfies Constraint~\dref{process}. Therefore, the optimal solution can be converted to a feasible solution to (LP1). This implies the desired inequality.
\end{proof}

Let $C_{ij}$ and $C_j$ denote the completion time of task $(i,j)$ and the completion time of job $j$ under \AI, respectively. In the next step for the proof of Theorem~\ref{theorem1}, we aim to bound the total volume of the first $j$ jobs (according to ordering~\dref{order}) that are processed during the time interval $(0,4 \tilde{C_j}]$ and subsequently use this result to bound $C_j$. \textit{Note that the list scheduling policy used in {\AI} is similar to the one used in {\AII}}, without the extra consideration for placement of fractions corresponding to the same task on different machines. Thus, The arguments here are similar to the ones in Lemma~\ref{tripledinterval}. Nevertheless, we present them for completeness.

Let $T_{ij}$ denote the first time that all the first $j$ tasks complete under {\AI} on machine $i$. Recall that, as a result of Constraint~\dref{process} and ordering in~\dref{order}, $\tilde{C_j} \geq \tilde{C_k} \geq p_{ik}$ for all $k \leq j$ and all $i \in \calM$. Further, the height of machine $i$ at time $t$ restricted to the first $j$ jobs is denoted by $h_{ij}(t)$ and defined as the height of machine $i$ at time $t$ when only considering the first $j$ jobs according to the ordering~\dref{order}. We have the following lemma.
\begin{lemma}
	\label{volumebound}
	Consider any interval $(\calT_1,\calT_2]$ for which $\calT_2-\calT_1=2\tilde{C_j}$ and suppose $\calT_2 < T_{ij}$ for some machine $i$. Then 
	\begin{equation} \label{eq:sumh}
	\sum_{t=\calT_1+1}^{\calT_2} h_{ij}(t) > m_i \tilde{C_j}
	\end{equation}
\end{lemma}
\begin{proof}[Proof of Lemma~\ref{volumebound}]
	\label{Proof1-3}
	Without loss of generality, consider interval $(0,2\tilde{C_j}]$ and assume $T_{ij} > 2\tilde{C_j}$.
	Let $S_{ij}(\tau)$ denote the set of tasks $(i,k)$, $k \leq j$ (according to ordering ~\dref{order}), running at time $\tau$ on machine $i$. We construct a bipartite graph $G=(U \cup V, E)$ as follows. For each time slot $\tau \in \{1,\dots,2\tilde{C_j}\}$ we consider a node $z_\tau$, and define $U= \{z_\tau| 1 \leq \tau \leq \tilde{C_j}\}$, and $V=\{z_\tau| \tilde{C_j}+1 \leq \tau \leq 2\tilde{C_j}\}$. For any $z_{s} \in U $ and $z_{t} \in V$, we add an edge $(z_s,z_t)$ if $S_{ij}(t) \setminus S_{ij}(s) \neq \varnothing$, i.e., there is a task $(i,k)$, $k \leq j$, running at time $t$ that is not running at time $s$. Note that existence of edge $(z_{s}, z_{t})$ implies that $h_{ij}(s)+h_{ij}(t) > m_i$, because otherwise {\AI} would have scheduled the task(s) in $S_{ij}(t) \setminus S_{ij}(s)$ (those that are running at $t$ but not at $s$) at time $s$. 
	
	Next, we show that a perfect matching of nodes in $U$ to nodes in $V$ always exists in $G$. The existence of perfect matching then implies that any time slot $s\in (0,\tilde{C_j}]$ can be matched to a time slot $t \in (\tilde{C_j},2 \tilde{C_j}]$ (one to one matching) and $h_{ij}(s)+h_{ij}(t) > m_i$. To prove that such a perfect matching always exists, we use Hall's Theorem~\cite{hall1935representatives}. For any set of nodes $\tilde{U} \subseteq U$, we define set of its neighbor nodes as $N_{\tilde{U}}=\{z_t \in V | \exists \ z_s \in \tilde{U}: (z_s,z_t) \in E\}$. 
	Hall's Theorem states that a perfect matching exists if and only if for any $\tilde{U} \subseteq U$ we have $|\tilde{U}| \leq |N_{\tilde{U}}|$, where $|\cdot|$ denotes set cardinality (size). 
	To arrive at a contradiction, suppose there is a (non-empty) set of nodes $\tilde{U} \subseteq U$ such that $|\tilde{U}| > |N_{\tilde{U}}|$. This implies that for a node $z_{t_1}$ in $V$ but not in the neighbor set of $\tilde{U}$, i.e., $z_{t_1} \in V \setminus N_{\tilde{U}}$, we should have  
	\begin{equation}
	\label{eq1}
	S_{ij}(t_1) \setminus S_{ij}(s) = \varnothing,
	\end{equation}
	for all $s, z_s \in \tilde{U}$. We now consider two cases:
	
	\textbf{Case (i)}: $|V \setminus N_{\tilde{U}}|=1$, which means $|N_{\tilde{U}}|=\tilde{C}_j-1$. But we had assumed $|\tilde{U}| > |N_{\tilde{U}}|$, thus $|\tilde{U}|=\tilde{C_j}$ and $\tilde{U}=U$. This implies that the tasks that are running at time $t_1$, are also running in the entire interval $(0,\tilde{C_j}]$; therefore, the processing time of each of them is at least $\tilde{C_j}+1$ which contradicts the fact that $\tilde{C_j} \geq p_{ik}$ for all jobs $k \leq j$, by Constraint~\dref{process} and ordering in~\dref{order}.
	
	\textbf{ Case (ii)}: $|V \setminus N_{\tilde{U}}| > 1$. In addition to the previous node $z_{t_1}$, consider another node $z_{t_2} \in V \setminus N_{\tilde{U}}$, and without loss of generality, assume $t_1<t_2$. Similarly to \dref{eq1}, it holds that 
	\begin{equation}
	\label{eq11}
	S_{ij}(t_2) \setminus S_{ij}(s) = \varnothing,
	\end{equation}
	for all $s, z_s \in \tilde{U}$. We claim that $S_{ij}(t_2) \subseteq S_{ij}(t_1)$, otherwise {\AI} would have moved some task $(i,k)$ running at $t_2$ and not at $t_1$ to time $t_1$ without violating machine $i$'s capacity. This is feasible because, in view of \dref{eq1} and \dref{eq11}, $(S_{ij}(t_1) \cup (i,k)) \setminus S_{ij}(s) = \varnothing$ for all $s, z_s \in \tilde{U}$. This implies that {\AI} has scheduled all tasks of the set $S_{ij}(t_1) \cup (i,k)$ simultaneously at some time slot $s\in (0,\tilde{C}_j]$, which in turn implies that adding task $(i,k)$ to time $t_1$ is indeed feasible (the total resource requirement of the tasks won't exceed $m_i$). Repeating the same argument for the sequence of nodes $z_{t_1}$, $z_{t_2}$, $\dots$, $z_{t_{|V \setminus N_{\tilde{U}}|}}$, where $t_1 < t_2 < \dots < t_{|V \setminus N_{\tilde{U}}|}$, we conclude that there exists a task that is running at all the times $t$, $z_t \in V \setminus N_{\tilde{U}}$, and at all the times $s, z_s \in \tilde{U}$. Therefore, its processing time is at least $\tilde{C_j}-|N_{\tilde{U}}|+|\tilde{U}|$ which is greater than $\tilde{C_j}$ by our assumption of $|\tilde{U}| > |N_{\tilde{U}}|$. This is a contradiction with the fact that $ p_{ik} \leq \tilde{C_j}$ for all $k \leq j$ by Constraint~\dref{process} and ordering~\dref{order}.
	
	Hence, we conclude that conditions of Hall's Theorem hold and a perfect matching in the constructed graph exists. As we argued, if $z_{s} \in U $ is matched to $z_{t} \in V$, we have $h_{ij}(s)+h_{ij}(t) > m_i$. Hence it follows that $\sum_{t=1}^{2\tilde{C_j}} h_{ij}(t) > m_i \tilde{C_j}$.
\end{proof}
Now we are ready to complete the proof of Theorem~\ref{theorem1} regarding the performance of \AI.
\begin{proof}[Proof of Theorem~\ref{theorem1}]
	Recall that $C_{ij}$ and $C_j$ denote completion time of task $(i,j)$ and completion time of job $j$ under \AI, respectively. Also, $T_{ij}$ denotes the first time that all the first $j$ tasks are completed under {\AI} on machine $i$. Therefore, $C_{ij} \leq T_{ij}$, by definition. 
	
	Define $i_j$ to be the machine for which $C_j=C_{i_jj}$. If $T_{ij} \leq 4 \tilde{C_j}$ for all machines $i \in \calM$ and all jobs $j \in \calJ$, we can then argue that $\sum_{j=1}^N w_j C_j \leq 4 \times \text{OPT}$, because
	\begin{equation*}
	\begin{aligned}
	\sum_{j=1}^N w_j C_j = \sum_{j=1}^N w_j C_{i_jj} \leq \sum_{j=1}^N w_j T_{i_jj}\stackrel{(a)} \leq 4 \sum_{j=1}^N w_j \tilde{C_j}\stackrel{(b)} \leq 4 \sum_{j=1}^N w_j C^\star_j,
	\end{aligned}
	\end{equation*}
	where Inequality (a) follows from our assumption that $T_{ij} \leq 4 \tilde{C_j}$, and Inequality (b) follows from Lemma~\ref{bound}. 	
	
	Now to arrive at a contradiction, suppose $T_{ij} > 4 \tilde{C_j}$ for some machine $i$ and job $j$. We then have,
	\begin{equation}
	\label{eq2}
	\begin{aligned}
	\sum_{k=1}^j v_{ik} &= \sum_{t=1}^{T_{ij}} h_{ij}(t) \stackrel{(c)} > \sum_{t=1}^{2\tilde{C_j}} h_{ij}(t)+ \sum_{t=1}^{2\tilde{C_j}} h_{ij}(t+2\tilde{C_j})\\
	&\stackrel{(d)} > m_i \tilde{C_j} + m_i \tilde{C_j}=2 m_i \tilde{C_j},
	\end{aligned}
	\end{equation}
	where Inequality (c) is due to the assumption that $T_{ij} > 4 \tilde{C_j}$, and Inequality (d) follows by applying Lemma~\ref{volumebound} twice, once for interval $(0,2\tilde{C_j}]$ and once for interval $(2\tilde{C_j},4\tilde{C_j}]$. But \dref{eq2} contradicts Lemma~\ref{ineq1}. Hence, $\sum_{j=1}^N w_j C_j \leq 4 \times \text{OPT}$.
\end{proof}

\section{Pseudocodes of  ($4$)-Approximation Algorithm}
\label{pseudocode2}
Algorithm~\ref{Alg1} provides a pseudocode for  {\AI}, our preemptive $4$-approximation algorithm, described in Section~\ref{specialcase}. The algorithm is a simple list scheduling based on the ordering obtained from (LP3).
\begin{algorithm}[h]
	\caption{Preemptive Scheduling Algorithm \AI}
	\label{Alg1}
	\begin{algorithmic}
		\State Given a set of machines $\calM = \{1, . . . , M\}$, a set of jobs $\calJ =\{1, . . . , N\}$, and weights $w_j$, $j \in \calJ$:
		\begin{algorithmic}[1]
			\State Solve (LP1) and denote its optimal solution by $\{\tilde{C_j}; j \in \calJ\}$.
			\State Order and re-index jobs such that $
			\tilde{C}_1 \leq \tilde{C}_2 \leq ... \leq \tilde{C}_N.$
			\State On each machine $i \in \calM$, apply list scheduling as described below:
			\While {There is some incomplete task,}
			\State List the incomplete tasks respecting the ordering in line 2. Let $Q$ be the total number of tasks in the list. Denote current time by $t$, and set $h_i(t)=0$ 
			\For {$q=1$ to $Q$}
			\State Denote the $q$-th task in the list by $(i,j_q)$
			\If {$h_i(t)+a_{ij_q} \leq m_i$,}
			\State Schedule task $(i,j_q)$.
			\State Update $h_i(t) \leftarrow h_i(t)+a_{ij_q}$.
			\EndIf
			\EndFor
			\State Process the tasks that were scheduled in line 9 until a task completes.
			\EndWhile
		\end{algorithmic}
	\end{algorithmic}
\end{algorithm}

\section{Pseudocodes of  ($6+\epsilon$)-Approximation Algorithm}
\label{pseudocode1}
A pseudocode for our preemptive ($6+\epsilon$)-approximation algorithm {\AII} described in Section~\ref{preemptivemigration} is given in Algorithm~\ref{Alg3}. Line 1 in Algorithm~\ref{Alg3} corresponds to Step 1 in Section~\ref{preemptivemigration}, lines 2-18 correspond to Step 2, construction of schedule $\calS$, and lines 19-20 describe Slow-Motion and construction of schedule $\bar{\calS}$ in Step 3.

\begin{algorithm}[H]
	\caption{Preemptive Scheduling Algorithm {\AII}}
	\label{Alg3}
	\begin{algorithmic}
		\State Given a set of machines $\calM = \{1, . . . , M\}$, a set of jobs $\calJ =\{1, . . . , N\}$, and weights $w_j$, $j \in \calJ$:
		\begin{algorithmic}[1]
			\State Solve (LP1) and denote its optimal solution by $\{\tilde{z}_{kj}^{il};\ j \in \calJ, \ k \in \calK_j, \ i \in \calM,\ l \in \{0,1,\dots,L\}\}$.
			\State List non-zero task fractions (i.e., tasks  $(k,j)$ with size $a_{ij}$ and non-zero fractional duration $\tilde{z}_{kj}^{il}p_{kj}^i$) such that task fraction $(k,j,i,l)$ appears before task fraction $(k^\prime, j^\prime,i^\prime,l^\prime)$, if $l<l^\prime$. Task fractions within each interval ando corresponding to different machines are ordered arbitrarily.
			\While {There is some unscheduled task fraction,}
			\State List the unscheduled task fractions. Let $Q$ be the total number of task fractions in the list and $l^\star$ be the interval with minimum value that has some unscheduled task fractions..
			\State  Denote current time by $t$, and set $h_i(t)$ to be the height of machine $i$ at $t$.
			\For {$q=1$ to $Q$}
			\State Denote the $q$-th task (whose fraction need to get scheduled) in the list by $(k_q,j_q)$.
			\If {$h_i(t)+a_{k_qj_q} \leq m_i$ and no fraction of task $(k_q,j_q)$ is running in any  other machine,}
			\State Schedule task $(k_q,j_q)$ to run on machine $i$.
			\State Update $h_i(t) \leftarrow h_i(t)+a_{k_qj_q}$.
			\EndIf
			\EndFor
			\State Process the task fractions that were scheduled in line 12 until a task fraction completes.
			\State Preempt scheduling of all the task fractions $(k,j,i,l)$ with $l>l^\star$.
			\State Update $\tilde{z}_{kj}^{il} \leftarrow \tilde{z}_{kj}^{il}-\tau/p_{kj}^i$, where $\tau$ is the amount of time it gets processed.
			\If {$\tilde{z}_{kj}^{il}=0$}
			\State Remove task $(k_q,j_q)$ from the list of machine $i$, and update $Q \leftarrow Q-1$.		
			\EndIf
			\EndWhile
			\State Denote the obtained schedule by $\calS$. Choose $\lambda$ randomly from $(0,1]$ with pdf $f(\lambda) = 2\lambda$.
			\State Construct schedule $\bar{\calS}$ by applying Slow-Motion with parameter $\lambda$ to $\calS$.
			Process jobs according to $\bar{\calS}$.
		\end{algorithmic}
	\end{algorithmic}
\end{algorithm}

\newpage
\section{Pseudocode of $24$-Approximation Algorithm}
\label{pseudocode3}
Algorithm~\ref{Alg4} provides a pseudocode for our non-preemptive algorithm, {\AIII}, described in Section~\ref{nomigration}. Line 1 in Algorithm~\ref{Alg4} corresponds to Step 1 in {\AIII} and lines 2 corresponds to Step 2, namely, construction of preemptive schedule and applying Slow-Motion. Lines 3-11 describes the procedure of constructing a non-preemptive schedule using $\bar{\calS}$ in Step 3.

Algorithm~\ref{subroutine} describes the mapping procedure which is used as a subroutine in Algorithm~\ref{Alg4}.

\begin{algorithm}[h]
	\caption{Non-Preemptive Scheduling Algorithm \AIII}
	\label{Alg4}
	\begin{algorithmic}
		\State Given a set of machines $\calM = \{1, . . . , M\}$, a set of jobs $\calJ =\{1, . . . , N\}$, and weights $w_j$, $j \in \calJ$:
		\begin{algorithmic}[1]
			\State Solve (LP2) and denote its optimal solution by $\{\tilde{z}_{kj}^{il},\ j \in \calJ, \ k \in \calK_j, \ i \in \calM,\ 0 \leq l \leq L\}$.
			\State Apply Slow-Motion by choosing $\lambda$ randomly from $(0,1]$ with pdf $f(\lambda) = 2\lambda$, and define $\bar{z}_{kj}^{il}$, as in \dref{zbardef}.
			\State Run Algorithm~\ref{subroutine} and output list of tasks that are mapped to each machine-interval $(i,l)$, $i \in \calM, \ l \leq L$.
			\For {Each machine $i \in \calM$,}
			\For {Each interval $l$, $0 \leq l \leq L$,}
			\State List the unscheduled task fractions. Let $Q$ be the total number of task fractions in the list.
				\While {There is some unscheduled task fraction,}
				\State Denote current time by $t$, and set $h_i(t)$ to be the height of machine $i$ at $t$.
				\For {$q=1$ to $Q$}
				\State Denote the $q$-th task in the list by $(k_q,j_q)$
				\If {$h_i(t)+a_{k_qj_q} \leq m_i$,}
				\State Schedule task $(k_q,j_q)$.
				\State Update $h_i(t) \leftarrow h_i(t)+a_{k_qj_q}$.
				\State Update $Q \leftarrow Q-1$.
				\EndIf
				\EndFor
				\State Process the task fractions that were scheduled in line 12 until a task fraction is complete.
				\EndWhile
			\EndFor
			\EndFor
		\end{algorithmic}
	\end{algorithmic}
\end{algorithm}

\newpage
\begin{algorithm}[h]
	\caption{Procedure of Mapping Tasks to Intervals}
	\label{subroutine}
	\begin{algorithmic}
		\State Given a set of jobs $\calJ =\{1, . . . , N\}$, with task volumes $v_{kj}^i$ on machine $i$, and values of $\bar{z}_{kj}^{il}$:
		\begin{algorithmic}[1]
			
			\State Construct bipartite graph $\calG_i=(U \cup V, E)$ as follows:
			\State For each task $(k,j)$, $j \in J, \ k \in \calK_j$, add a node in $U$.
			\For {Each machine $i$, $i \in \calM$,}
			\State Order and re-index tasks such that:
			$
			v_{k_1j_1}^i \geq v_{k_2j_2}^i \geq \dots v_{k_{N_i}j_{N_i}}^i >0.
			$
			\For {Each interval $l$, $l \leq L$,}
			\State Consider $\lceil \bar{z}^{il} \rceil=\lceil \sum_{j \in J} \sum_{k \in \calK_j} \bar{z}_{kj}^{il} \rceil$ consecutive nodes in $V_i$, and set $W_l^{ic_l}=0$ for $1 \leq c_l \leq \lceil \bar{z}^{il} \rceil$. Also set $c_l=1$.
			\For {$q=1$ to $N_i$}
			\State $R=\bar{z}_{kj}^{il}$,
			\While {$R \neq 0$}
			\State Add an edge between the node $(k_q,j_q)$ in set $U$ and node $c_l\in V_i$.
			\State Assign weight $w_{kj}^{ilc}=\min\{R,1-W_l^{c_l}\}$.
			\State Update $R \leftarrow R-w_{kj}^{ilc}$.
			\State Update $W_l^{c_l} \leftarrow W_l^{c_l}+w_{kj}^{ilc}$
			\If {$W_l^{c_l}=1$}
			\State $c_l=c_l+1$.
			\EndIf
			\EndWhile
			\EndFor
			\EndFor
			\EndFor
			\State Set $V=\cup_{i \in \calM} V_i$.
			
			\State Find an integral matching in $\calG$ on the nonzero edges with value $|\cup_{j \in \calJ} \calK_j|=\sum_{j \in \calJ} \sum_{k \in \calK_j} \sum_{l=0}^L \bar{z}_{kj}^{il}$.
		\end{algorithmic}
	\end{algorithmic}
\end{algorithm}

\end{document}